\documentclass{article}

\usepackage{microtype}
\usepackage{graphicx}
\usepackage{subfigure}
\usepackage{booktabs} % for professional tables

\usepackage{hyperref}
\usepackage{multirow}
\usepackage{float}

\usepackage[accepted]{icml2023}

\usepackage{amsmath}
\usepackage{amssymb}
\usepackage{mathtools}
\usepackage{amsthm}

\usepackage{algorithm}
\usepackage{algorithmic}

\usepackage[capitalize,noabbrev]{cleveref}

\theoremstyle{plain}
\newtheorem{theorem}{Theorem}[section]

\newtheorem{lemma}[theorem]{Lemma}

\theoremstyle{definition}

\newtheorem{assumption}[theorem]{Assumption}
\theoremstyle{remark}

\usepackage[textsize=tiny]{todonotes}

\icmltitlerunning{Neural Attention Additive Q-learning}

\begin{document}

\twocolumn[
\icmltitle{N$\text{A}^\text{2}$Q: Neural Attention Additive Model for Interpretable\\ Multi-Agent Q-Learning}




\begin{icmlauthorlist}
	\icmlauthor{Zichuan Liu}{yyy}
	\icmlauthor{Yuanyang Zhu}{yyy}
	\icmlauthor{Chunlin Chen}{yyy}
	%\icmlauthor{Firstname4 Lastname4}{equal,sch}
	%\icmlauthor{Firstname5 Lastname5}{equal,yyy}
	%\icmlauthor{Firstname6 Lastname6}{equal,sch,yyy,comp}
	%\icmlauthor{Firstname7 Lastname7}{comp}
	%\icmlauthor{Firstname8 Lastname8}{sch}
	%\icmlauthor{Firstname8 Lastname8}{yyy,comp}
\end{icmlauthorlist}
%\icmlaffiliation{yyy}{Nanjing University, Nanjing, China}
\icmlaffiliation{yyy}{Department of Control Science and Intelligence Engineering, Nanjing University, Nanjing, China}
\icmlcorrespondingauthor{Yuanyang Zhu}{yuanyang@smail.nju.edu.cn}
\icmlcorrespondingauthor{Chunlin Chen}{clchen@nju.edu.cn}

\icmlkeywords{Machine Learning, ICML}

\vskip 0.3in
]



\printAffiliationsAndNotice{}  % leave blank if no need to mention equal contribution

\begin{abstract} 
	Value decomposition is widely used in cooperative multi-agent reinforcement learning, however, its implicit credit assignment mechanism is not yet fully understood due to black-box networks.
	In this work, we study an interpretable value decomposition framework via the family of generalized additive models.
	We present a novel method, named Neural Attention Additive Q-learning (N$\text{A}^\text{2}$Q), providing inherent intelligibility of collaboration behavior.
	N$\text{A}^\text{2}$Q can explicitly factorize the optimal joint policy induced by enriching shape functions to model all possible coalitions of agents into individual policies. 
	Moreover, we construct identity semantics to promote estimating credits together with the global state and individual value functions, where local semantic masks help us diagnose whether each agent captures relevant-task information.
	Extensive experiments show that N$\text{A}^\text{2}$Q consistently achieves superior performance compared to different state-of-the-art methods on all challenging tasks, while yielding human-like interpretability.
	% Source code will be released as soon as possible.
	%	We also demonstrate the contribution of agents to interactions as an explanation
	%	Source code is placed on  \url{https://anonymous.4open.science/r/NA2Q-2I0C2M3L}.
\end{abstract}

\section{Introduction}

\begin{figure}[ht]
	% \vskip 0.2in
	\begin{center}
		\centerline{\includegraphics[width=1\linewidth]{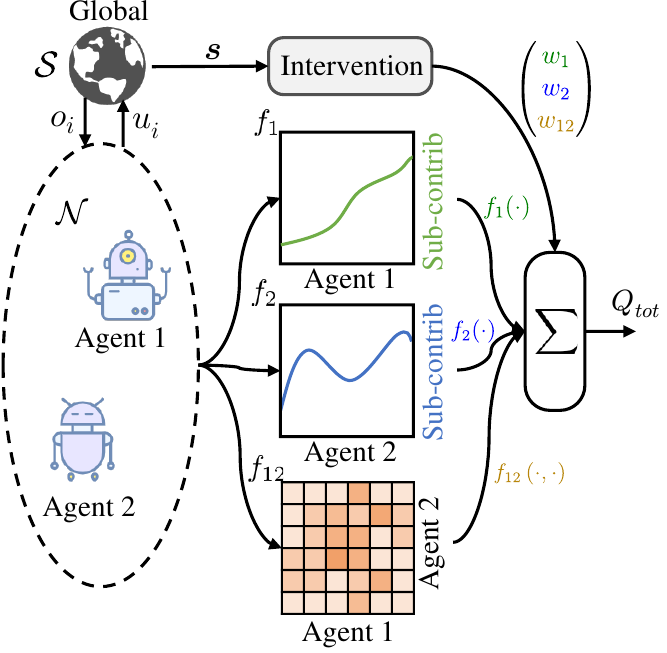}}
		\caption{An example of value decomposition via the GAMs family in MARL, where $\boldsymbol{s}\in \mathcal{S}$ is the global state, $f_k\in \{f_1, \cdots, f_{1 \ldots n}\}$ denotes the contribution of a shape function to learning individual or pairwise action values, and $Q_{tot}$ denotes the joint action value. 
		}
		\label{Introduction}
	\end{center}
	\vskip -0.3in
\end{figure}

Cooperative multi-agent reinforcement learning (MARL) has been proven to hold considerable promise for addressing many challenging real-world problems, e.g., autonomous driving~\cite{kiran2021deep}, scene understanding~\cite{chen2019counterfactual}, and robotics~\cite{kober2013reinforcement, lillicrap2015continuous}.
Value decomposition~\cite{rashid2018qmix,son2019qtran,wang2020qplex} has witnessed success in handling the joint action-value function effectively in value-based MARL methods.
This progress has been fueled by black-box neural structures, where the underlying decision process and credit assignment mechanisms are difficult for humans to understand and interpret.
Hence, explicitly understanding the decision-making processes and deducing the contribution of agents is still crucial in the MARL community.

A growing body of work attempts to demystify the decision-making process of deep reinforcement learning.
Instance-specific approximation methods aim to explain black-box predictions via the Shapley value~\cite{wang2020shapley} or clustering~\cite{zahavy2016graying} techniques in post-hoc explanation techniques.
However, these interpretable methods are considered computationally expensive~\cite{slack2021reliable} and unstable~\cite{ghorbani2019interpretation}, i.e., they often misrepresent models or agents' decisions.
Other works~\cite{bastani2018verifiable, silva2020optimization} have resorted to imitation learning to generate post-hoc global explanations aimed at distilling agent strategies, which lack the transparency of the original model and do not guarantee performance in complex tasks.
This landscape has ignited interest in intrinsic explanations, particularly in generalized additive models (GAMs)~\cite{hastie1986generalized}.
GAMs typically learn independent shape functions for each feature, whose outputs are combined for the final prediction, ensuring isolated contributions, e.g., NIT~\cite{tsang2018neural}, NAM~\cite{agarwal2021neural}, and NODE-GAM~\cite{chang2021node}.
Moreover, it can model all higher-order feature interactions with expressive power and easy scalability.
These successful interpretable GAMs stimulate our thinking in MARL domains, i.e., could GAMs facilitate more trustworthy agent collaboration and efficient credit assignment?

To leverage the benefits of GAM in MARL effectively, we introduce a unique value decomposition workflow as illustrated in in~\textit{Figure~\ref{Introduction}}. 
At each timestep $t$, each agent learns the decentralized action-value $Q_i$ and passes it to the central GAM while executing the action $u_i$, which then evaluates the team utility via the joint action-value $Q_{tot}$. 
Specifically, our GAM consists of several independent shape functions, where each function inputs a marginal or higher-order action value, outputting the corresponding agent's team contribution.
We restrict our attention to unary and pairwise shape functions to maintain interpretability and efficiently infer credits, helping in isolating individual and pairwise coalition contributions. 
However, the causal confounder is correlated with $\boldsymbol{s}$ and $Q_{tot}$, creating a spurious correlation among them, potentially complicating the learning of correct causal relationships. 
Drawing inspiration by~\cite{glymour2016causal,li2022deconfounded}, to relieve the spurious correlation between $\boldsymbol{s}$ and $Q_{tot}$, we construct local semantics alongside the global state to compute credits. 
In this case, it explicitly provides a perspective on diagnosing whether the individual agent could effectively avoid the negative influence of focusing on irrelevant input information.
Meanwhile, this brings about semantic masks that can diagnose agents' local observations.
We utilize the attention mechanism~\cite{vaswani2017attention} as an intervention term to capture the credit of each shape function, facilitating effectively capturing credit assignment. 
We call this comprehensive solution Neural Attention Additive Q-learning~(N$\text{A}^\text{2}$Q), which offers a fresh perspective for interpreting collaboration among agents and understanding local semantics.
%However, a causal confounder exists, correlating $\boldsymbol{s}$ and $Q_{tot}$, causing spurious correlations that potentially impede the learning of proper causal relationships.
%Drawing inspiration from~\cite{glymour2016causal,li2022deconfounded}, we construct local semantics alongside the global state to compute credits.
%This process also generates semantic masks aiding in agent observation diagnosis, and employs an attention mechanism~\cite{vaswani2017attention} for capturing shape function credit, facilitating effective credit assignment.
%We term this comprehensive solution Neural Attention Additive Q-learning~(N$\text{A}^\text{2}$Q), which offers a fresh perspective for interpreting collaboration among agents and understanding local semantics.

Our contributions are summarized as follows: 
(1) We propose a novel value decomposition method, called Neural Attention Additive Q-learning (N$\text{A}^\text{2}$Q), which moves a step towards modeling all possible higher-order interactions and interpreting their collaboration behavior. 
We give rigorous proof that N$\text{A}^\text{2}$Q guarantees an acceptable regret bound by enriching the Taylor expansion of $Q_{tot}$ based on the GAM family.
(2) We provide diagnostic insights into what the agent captured from its observation by maximizing the observation resemblance and generating masks through encoding the local semantics, which is applied to the mixer to promote credit deduction.
(3) Through extensive experiments on challenging MARL benchmarks, N$\text{A}^\text{2}$Q not only consistently achieves superior performance compared to different state-of-the-art methods but also allows for an easy-to-understand of credit assignment among agents.

\section{Preliminaries}
\subsection{Dec-POMDP}
A fully cooperative multi-agent task generally can be formulated as a Dec-POMDP~\cite{oliehoek2016concise}, which consists of a tuple $\left< \mathcal{N}, \mathcal{S}, \mathcal{U}, \mathcal{P}, r, O, \Omega, \gamma  \right>$, where $\mathcal{N}$ represents a finite set of $n$ agents, and $\boldsymbol{s}\in \mathcal{S}$ describes the global state of the environment.
At each time step, each agent $i\in \mathcal{N}$ receives its own observation $o_i\in \Omega $
according to the partial observation $O(\boldsymbol{s}, i)$ and chooses an action $u_i\in \mathcal{U}$ to formulate a joint action $\boldsymbol{u}=[u_i]_{i=1}^n\in \mathcal{U}^n$. It results in a next state transition $\boldsymbol{s}'$ according to the transition function $\mathcal{P}(\boldsymbol{s}'|\boldsymbol{s}, \boldsymbol{u}): \mathcal{S}\times \mathcal{U}^n\to \mathcal{S}$ and all agents receive a joint reward $r(\boldsymbol{s}, \boldsymbol{u}): \mathcal{S}\times \mathcal{U}^n\to \mathbb{R}$. 
Moreover, each agent $i$ learns its own policy $\pi_i(u_i|\tau_i):\mathcal{T}\times \mathcal{U}\to [0, 1]$ conditions on its local action-observation history $\tau_i\in \mathcal{T}$, and we define $\boldsymbol{\tau}\in \boldsymbol{\mathcal{T}}$ to denote joint action-observation history.
The formal goal of all agents is to maximize the joint value function $Q^{\boldsymbol{\pi}} =  \mathbb{E}\left[\sum_{t=0}^{\infty} \gamma^{t} r^{t}\right]$ that finds an optimal joint policy $\boldsymbol{\pi}=\left[\pi_{i}\right]_{i=1}^n$, where $\gamma\in [0, 1)$ is a discount factor.

\subsection{Credit Assignment in MARL}
Value decomposition methods by credit assignment~\cite{sunehag2017value,rashid2018qmix,wang2020qplex} are the most popular branches in the centralized training and decentralized execution (CTDE)~\cite{oliehoek2008optimal} paradigm.
These methods should satisfy the individual-global-max~(IGM) principle~\cite{son2019qtran} to guarantee the consistency between local and global greedy actions as
\begin{equation}
\arg \underset{\boldsymbol{u}\in \mathcal{U}^n}{\max}\,Q_{tot}(\boldsymbol{\tau},\boldsymbol{u}) = \begin{pmatrix}
\arg \underset{u_1\in \mathcal{U}}{\max}\, Q_{1}(\tau _1,  u_1) \\
\vdots  \\
\arg\underset{u_n\in \mathcal{U}}{\max}\, Q_{n}(\tau_n, u_n)
\end{pmatrix},
\label{igm}
\end{equation}
where $Q_{tot}\in \mathcal{Y}$ is the joint action value for each individual value function $Q_i(\tau_i, u_i)$. 
Under this principle, credit assignment aims to infer the contributions of predecessor value functions to $Q_{tot}$~\cite{li2022deconfounded}. 
The decomposition values $[Q_i]_{i=1}^n\in \mathcal{Q}$ are usually transformed into temporal values $[\widehat{Q}_k]_{k=1}^m$ via a human-designed function $f_k$ with the global state $\boldsymbol{s}$, where $m$ is the function number.
It can represent a more general formulation as $Q_{tot} = \sum_{k=1}^m \alpha_{k} \widehat{Q}_k$ with the credit $\alpha_{k}$, and we assume $[Q_i]_{i=1}^n$ and $Q_{tot}$ are drawn following a kind of fixed (but unknown) distribution $\mathfrak{P}: \mathcal{Q}\to  \mathcal{Y}$.
The introduction of representative algorithms for the above formulation can be referred to in Appendix~\ref{appendix1}.

%\section{Neural Attention Additive Q-learning}
%In this section, we introduce Neural Attention Additive Q-learning (N$\text{A}^\text{2}$Q) that decomposes each shape function $f_i$ into a small set of basis functions shared among all the agents.
%It can exactly model the contribution of any agent, and explicitly present their correlation with each other.
%Besides, to better realizes deconfounded credit assignment, we leverage the trajectories $\tau$ of each agent to capture the individual semantic via backdoor adjustment, which also helps us understand the relative importance of different task-relevant observations during the decision-making process in a more interpretable way.
%%.extending the joint action-value function in terms of $Q_i$ by the Taylor expansion and 

\section{Theoretical Analysis for Decomposition}
\label{sec4}
Previous value-based studies have achieved great success in handling the joint action-value function to effectively enable CTDE in MARL.
However, they often suffer from at least one of the following limitations: 
(1) Value decomposition ideas~\cite{son2019qtran, wang2020qplex, iqbal2021randomized,rashid2020weighted} with complex non-linear transformations may often fail to allow us to explicitly model the contribution of each agent or coalition of agents.
(2) VDN~\cite{sunehag2017value}, Qatten~\cite{yang2020qatten}, and SHAQ~\cite{wang2021shaq} measure the importance of each individual to the team, which ignores potentially all possible coalitions of all agents.
(3) Existing mainstream value decomposition methods seldom diagnose whether individual agents can focus on specific information to help the mixer build a reasonable correlation of credit assignment between the global state $\boldsymbol{s}$ and the joint value function with a limited view of their surroundings.
To resolve these problems, we propose a novel interpretable value decomposition method that uses the neural additive model to learn higher-order permutation relationships of each agent in terms of the local expansion of the joint action-value $Q_{tot}$, which achieves a better trade-off between performance and interpretability.

Following the general framework of the value decomposition method, we recall the joint action-value function and expand it in terms of $Q_i$ by the Taylor expansion as
\begin{equation}
Q_{tot}=f_0+\sum_{i=1}^n\alpha_{i}Q_i+\cdots+\sum_{i_1,\ldots, i_l}^n \alpha_{i_1\ldots i_l} \prod_{j=1}^l Q_{i_j}+\cdots,
\label{eqqqq}
\end{equation}
where $f_0$ is a constant, all partial derivatives $\alpha_i = \frac{\partial Q_{tot}}{\partial Q_i}$ of order-$1$, and $\alpha_{i_1\ldots i_l}=\frac{1}{l!}\frac{\partial^l Q_{tot}}{\partial Q_{i_1}\ldots \partial Q_{i_l}}$ of order-$l$.
In this term, it can be seen as a simple polynomial GAM expression~\cite{dubey2022scalable} with full $n$ order interactions, which theoretically allows for learning any possible interaction order relationship among all agents.
We enrich Eq.~(\ref{eqqqq}) with a general neural additive model (NAM)~\cite{agarwal2021neural}, as an extended GAM method, providing more precise predictions for the contribution of individual agents and coalitions of agents, which is formulated as
\begin{equation}
\begin{split}
Q_{tot}=f_0 +\overbrace{\sum_{i=1}^n \alpha_{i} \underbrace{f_i\left(Q_i\right)}_{\text {order-}1}}^{\normalsize{\textcircled{\scriptsize{1}}}\normalsize \text{ similar to VDN}}+\cdots+
\sum_{k\in\mathcal{D}_l} \alpha_{k} \underbrace{f_{k}\left(Q_k\right)}_{\text {order-}l}\\
+\cdots+\overbrace{\alpha_{1 \ldots n}\underbrace{f_{1 \ldots n}(Q_1,\ldots,Q_n)}_{\text {order-}n}}^{\normalsize{\textcircled{\scriptsize{2}}}\normalsize \text{ e.g., QMIX}},
\end{split}
\label{sdada}
\end{equation}
%where $f_k \in \{f_1, \cdots, f_{1 \ldots n}\}^m$ is possibly a non-linear transformation to represent a shape function that obtains an intermediate value as $\widehat{Q}_k=f_k(\cdot)$, 
where $f_k \in \{f_1, \cdots, f_{1 \ldots n}\}^m$ is a shape function that transforms $l$ local values $Q_k$ into a temporal value $\widehat{Q}_k$,
and $\mathcal{D}_l$ is the set of all non-empty subsets of $l \in \{1, \cdots, n\}$ with order-$l$ interactions, i.e., $D_l = \{i_1\ldots i_l\}$.
When searching for a better value decomposition, we are often interested in this enrichment of the difference.
%Thus, we consider bounds on the empirical risk $\mathcal{L}(\widehat{Q}_{tot})$ in Eq~(\ref{eqqqq}) substituting the expected risk $\mathcal{L}(Q_{tot}^{\star})$ in Eq~(\ref{eqqqq}) as a regret bound.
To this end, we introduce the empirical risk minimizer $\widehat{Q}_{tot}$ in Eq~(\ref{sdada}) and the expected risk minimizer $Q_{tot}^{\star}$ in Eq~(\ref{eqqqq}) and consider $\mathcal{L}(\widehat{Q}_{tot})-\mathcal{L}(Q_{tot}^{\star})$ as a regret bound.
The conclusion shows that an upper bound always exists on our generalization according to regret analysis under the \textit{1-Lipschitz} loss approximation.
We provide approximation guarantees and detailed derivations for this type of enrichment, which can be found in Appendix~\ref{app2} along with rigorous proofs.
%To demonstrate that learning the rich decomposition method in Eq~(\ref{sdada}) does not produce a larger error compared to the Taylor expansion in Eq~(\ref{eqqqq}), we consider bounds on empirical risk $\mathcal{L}(\widehat{Q}_{tot})$ substituting for expected risk $\mathcal{L}(Q_{tot}^{\star})$.
% and give our error upper bound under the \textit{1-Lipschitz} loss approximation.
%We provide learning-theoretic and approximation guarantees for this type of enrichment, with a more precisely excess risk under regret bound,  which can be found in Appendix~\ref{app2} along with rigorous proof.

Most existing MARL methods primarily focus on one of the terms in Eq.~(\ref{sdada}), aiming to maximize performance while neglecting the different orders of coalition among agents.
For instance, VDN decomposes $Q_{tot}$ into a sum of individual action values representing only a limited class with order-$1$, i.e., it is similar to term $\normalsize{\textcircled{\scriptsize{1}}}\normalsize$ with equal credits.
QMIX considers mixing all individual action values as the most effective value decomposition method falling under term $\normalsize{\textcircled{\scriptsize{2}}}\normalsize$, but it does not explicitly show its credit assignment.
%When $f_i(Q_i)=Q_i$ and $\alpha_i=1$ in term $\normalsize{\textcircled{\scriptsize{1}}}\normalsize$, this term implies equal but failed credit assignment like VDN. 
%Term $\normalsize{\textcircled{\scriptsize{2}}}\normalsize$ considers mixing all individual action-values as the most effective value decomposition method, e.g., QMIX, but it could not explicitly show its credit assignment.
It is widely recognized~\cite{lou2013accurate, chang2021node} that it ceases to be interpretable with increasing order in Eq.~(\ref{sdada}), e.g., beyond pairwise interactions, albeit with some advantages to performance. 
Following this idea, we aim to maintain both the performance and interpretability of collaboration relationships in terms of any order of interaction by learning each unary and pairwise shape function as
%As a result, we consider training for a low-order decomposition in our experiments by learning each unary and pairwise shape function as
\begin{equation}
Q_{tot} = f_0 +\sum_{i=1}^{n} \alpha_i f_i(Q_i)+\sum_{ij\in \mathcal{D}_2} \alpha_{ij} f_{ij}(Q_i, Q_j). \label{eq33333}
\end{equation}

\begin{figure*}[ht]
	%	\vskip 0.2in
	\begin{center}
		\centerline{\includegraphics[width=\linewidth]{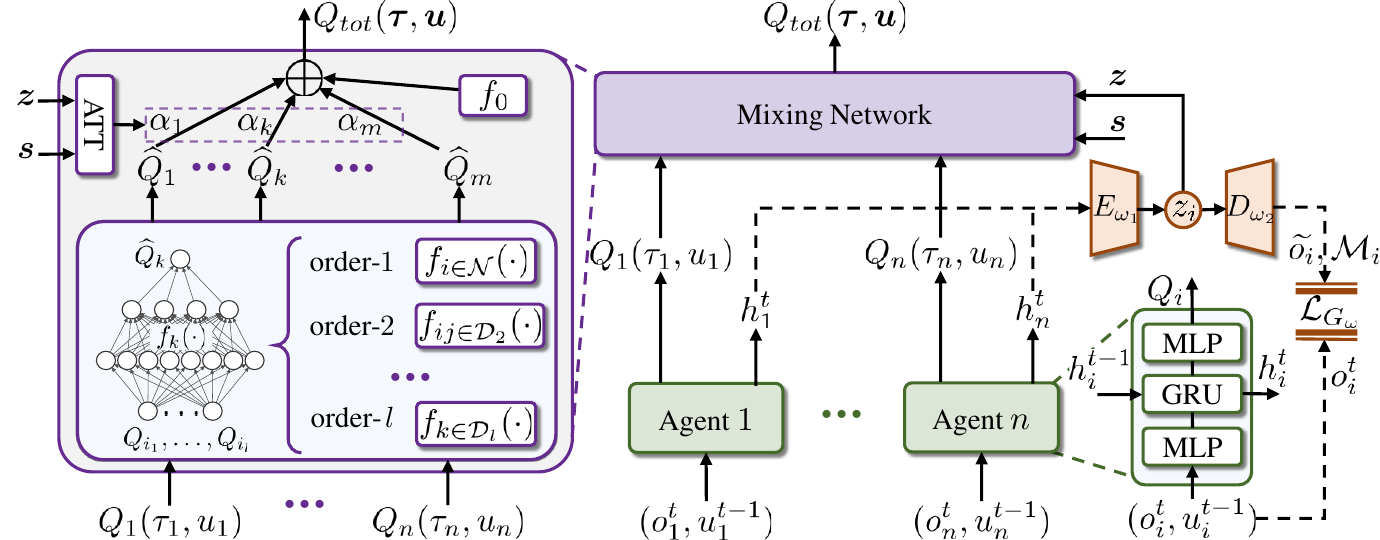}}
		\vskip -0.05in
		\caption{The overall framework of N$\text{A}^\text{2}$Q.
			First, each agent receives its local action-observation history $\tau_i$ and models its individual value function $Q_i(\tau_i, u_i)$.
			Next, we construct the identity semantic $z_i$ by encoding $\tau_i$, and take it together with the global state $\boldsymbol{s}$ to estimate credits, which provides a captured semantic interpretation. 
			In the mixing network, we transform the local Q-values $[Q_i]_{i=1}^n$ into temporal Q-values $[\widehat{Q}_k]_{k=1}^m$ by the shape function $f_k$ within order-$l$, where $l\in \mathcal{N}$, which are used to predict the joint Q-value together with credits.
		}
		\label{111111111}
	\end{center}
	\vskip -0.25in
\end{figure*}

Furthermore, previous works~\cite{rashid2018qmix, iqbal2021randomized} typically learn credit $\alpha_k$ by the global state $\boldsymbol{s}$, with $P(Q_{tot}| \boldsymbol{s})$ calculated. 
It brings a spurious association between $\boldsymbol{s}$ and $Q_{tot}$ that restricts deducing the contributions of individual agents and sub-teams from the overall success~\cite{li2022deconfounded}.
One possible solution is to impose an intervention function on $\boldsymbol{s}$ by identifying the local history $\tau_i$ in an unobservable environment.
Moreover, we are unsure whether the individual agent captures the important information that could help the mixer produce the credits from its observation instead of blindly pursuing performance.
To this end, from the perspective of diagnosing the individual agent, we explicitly generate an individual semantic $z_i$ from  $\tau_i$ to achieve the identity representation, and then decentralized credit assignment is obtained by calculating 
\begin{equation}
P(Q_{tot}| I(\boldsymbol{s})) = \sum _{\boldsymbol{z}} P(Q_{tot}| \boldsymbol{s}, \boldsymbol{z}) P(\boldsymbol{z}),
\label{model}
\end{equation}
where $I(\cdot)$ denotes the intervention function and the joint semantic $\boldsymbol{z}$ is generated for backdoor adjustment by sampling individual semantics as $\boldsymbol{z}=[z_i \sim  P(\tau_i)]^n_{i=1}$.
It helps us diagnose whether individual agents focus on the relative importance of different task-relevant observations during the decision-making process in a more interpretable manner.
%where $I(\cdot)$ denotes the intervention function and the joint semantic $\boldsymbol{z}$ is generated via backdoor adjustment  by sampling individual semantic as $\boldsymbol{z}=[z_i \sim  P(\tau_i)]^n_{i=1}$.
%It helps us understand the relative importance of different task-relevant observations during the decision-making process in a more interpretable way.

\section{Neural Attention Additive Q-learning}
Based on the previous analysis in Section~\ref{sec4}, we propose a novel interpretable value decomposition method, called \textit{Neural Attention Additive Q-learning}~(N$\text{A}^\text{2}$Q), that explicitly learns a decomposition mapping for all possible order interactions among agents and captures semantic information from their observations.
\textit{Figure~\ref{111111111}} illustrates the overall training procedure.
For each agent, N$\text{A}^\text{2}$Q models a local value function $Q_i(\tau_i, u_i)$ and generates the identity semantic $z_i$ by encoding the history $\tau_i$. 
In this process, we maximize the resemblance of observations through decoding to ensure the accuracy of semantic information and upsample masks as an interpretation.
%In this case, it explicitly provides a perspective of diagnosing whether the individual agent could effectively avoid the negative influence of irrelevant input information.
In the mixer, the local Q-values $[Q_i]_{i=1}^n$ are transformed  into $[\widehat{Q}_k]_{k=1}^m$ by all interactions of shape functions within order-$l$ among agents, and the united semantics $\boldsymbol{z}$ and the global state $\boldsymbol{s}$ are fed into the intervention function to estimate credits.
The joint value function is predicted depending on the temporal values $[\widehat{Q}_k]_{k=1}^m$ as well as credits.
It can exactly model the contribution of any agent or coalition of agents to the overall success by enriching Eq.~(\ref{eqqqq}) with NAM.

%According to the foregoing analysis, we propose the \textit{Neural Attention Additive Q-learning}~(N$\text{A}^\text{2}$Q) following the decomposition formation in Eq.~(\ref{sdada}) that decomposes each shape function $f_k$ into a small set of basis functions shared among all the agents.
%It can exactly model the contribution of any agent, and explicitly present their correlation with each other.
%\textit{Figure~\ref{111111111}} illustrates the overall training procedure.
%First, each agent models an individual value function $Q_i(\tau_i, u_i)$ and generates identity semantic $z_i$ via encoding the trajectories $\tau_i$.
%Afterward, the contextual semantic $\boldsymbol{z}$ is united by all agents, and the global state $\boldsymbol{s}$ is used as better realized shape function credits.
%Meanwhile, we maximize observations' resemblance by decoding to ensure the accuracy of semantic information and generate masks as interpretation.
%Finally, all shape functions with transforming individual Q-values are weighted within order-$l$ by credits to the joint value.

\textbf{Individual Action-Value Function.} 
Following the mainstream works~\cite{rashid2018qmix, wang2020qplex}, we employ a recurrent Q-function~\cite{hausknecht2015deep} with parameter sharing for each agent $i$. 
Specifically, each function takes current local observation $o^t_i$ with previous action $u^{t-1}_i$ and previous hidden state $h^{t-1}_i$ as inputs, and then outputs current hidden state $h^t_i$ and local Q-value.

\textbf{Constructing Identity Semantic.} 
For each agent, we consider a general setting in which each agent focuses its observation on task-relevant regions.
To capture this focus, we construct an underlying latent semantic from a local action observation of each agent via a variational auto-encoder (VAE)~\cite{sohn2015learning}, which can produce semantics corresponding to the importance assigned to each input observation.
Over the course of training, the action-observation $\tau_i$ of each agent $i$ is encoded by the VAE $G_{\omega} = \left\{E_{\omega_1},  D_{\omega_2} \right\}$ to sample its own identity semantic as $z_i =  E_{\omega_1}(\tau_i)$.
This semantic is then used as input for $D_{\omega_2}$ and upsampled to generate an attention mask as $\mathcal{M}_i =  \varsigma (D_{\omega_2}(z_i))$, where $\varsigma (\cdot)$ represents the sigmoid function. 
Generally, the mask is interpreted to show where the agent is ``looking" to make a decision~\cite{shi2020self}. 
To maximize the resemblance between the identity semantic $z_i$ and local observation $o_i$, the VAE $G_{\omega}$ is trained on a loss of the reconstruction observation along with a KL-divergence as
\begin{equation*}
\mathcal{L}_{vae}=\sum_{i=1}^n \left\| o_i-\widetilde{o_i}\right\|^2_2 + D_{\text{KL}}(\mathcal{N}(\mu, \sigma)||\mathcal{N}(0, I)), \label{eq2}
\end{equation*}
where $\widetilde{o_i} = \mathcal{M}_i\odot o_i$ and $\odot$ represent the overlaid observation with the mask and the element-wise multiplication, respectively.
%The outputs of encoder $E_{\omega_1}$ are the mean $\mu$ and standard deviation $\sigma$ of a Gaussian distribution $\mathcal{N}(\mu, \sigma)$, whose introduction can be found in Appendix~\ref{app3}.
The normal distribution $\mathcal{N}(\mu, \sigma)$ is represented by deterministic functions, whose introduction is deferred to Appendix~\ref{app3}.
Meanwhile, the mask is expected to focus on as sparse and relevant region information as possible, so we apply a direct penalty to the mask by $L_1$-norm as
\begin{equation}
\mathcal{L}_{G_{\omega}}=\mathcal{L}_{vae}+ \sum_{i=1}^n\left\|  \mathcal{M}_i \right\|_1.
\label{eq3}
\end{equation}
By training the VAE with parameters $\omega = \{\omega_1, \omega_2\}$, we can obtain an attention mask to help humans better understand the local observation and latent identity semantic of each agent to influence the prediction of its action.

\textbf{Learning Decomposition with Credit Assignment.} 
To accomplish the decomposition formation in Eq.~(\ref{eq33333}) and (\ref{model}), we let  $Q_{tot}$ be decomposed into a neural GAM paradigm within order-$2$ by setting $l\leq 2$ as
\begin{equation}
Q_{tot} = f_0(\boldsymbol{s}) +\sum_{i=1}^{n} \alpha_i f_i(Q_i)+\sum_{i=1}^{n}\sum_{j>i}^{n} \alpha_{ij} f_{ij}(Q_i, Q_j), \label{eq5}
\end{equation}
where $f_0$ is a bias term, univariate and bivariate shape functions $f_k$ are nonlinear functions (e.g., lightweight MLPs). 
To satisfy the IGM principle in Eq.~(\ref{igm}), we restrict all the network weights to be non-negative by using the absolute in $f_k$.
Considering that more efficient credit assignments can help local agents predict their actions more precisely, we also introduce the intervention function to realize decomposed training for backdoor adjustment.
Specifically, the credit $\alpha_{k}$ is computed with the identity semantics $[z_i]_{i=1}^n$ and the global state $\boldsymbol{s}$ through a dot-product attention as
\begin{equation}
\alpha_{k} = \left [\alpha_{i}, \alpha_{ij}\right ] = \frac{\exp((\boldsymbol{w}_z \boldsymbol{z})^\top \text{ReLU}(\boldsymbol{w}_s \boldsymbol{s}))}{\sum_{k=1}^m \exp((\boldsymbol{w}_z \boldsymbol{z})^\top \text{ReLU}(\boldsymbol{w}_s \boldsymbol{s}))
}, \label{eq4}
\end{equation}
where $\boldsymbol{w}_s$, $\boldsymbol{w}_z$ are the learnable parameters, and ReLU is employed as the activation function.  
$\alpha_{k}$ is positive with softmax operation to ensure monotonicity.

\textbf{Interpretability.} 
Interpreting decomposition in Eq.~(\ref{eq5}) is intuitive as the influence of an individual Q-value on the prediction operates independently of other action values. 
It is possible to visualize the mapping relationships by visualizing the univariate shape function $f_i$, e.g., plotting $Q_i$ on the $x$-axis and  $\alpha_{i}f_i(Q_i)$ on the $y$-axis. 
The bivariate shape function $f_{ij}$ is visualized through a heatmap~\cite{lou2013accurate, radenovic2022neural}, which is commonly used to achieve interpretation. 
Note that the visualization of the function accurately depicts how N$\text{A}^\text{2}$Q computes a prediction.
In addition, the semantics of individual agents are upsampled into masks to represent feature importance, increasing the confidence of local observations on the semantics.

The overall learning objective is to end-to-end train the whole framework by minimizing the loss $\mathcal{L}$ with the mean squared temporal-difference (TD) error as 
\begin{equation}
\mathcal{L}(\theta, \omega)=
\left\| Q_{tot}(\boldsymbol{\tau}, \boldsymbol{u}) -y\right\|^2_2 +\beta \mathcal{L}_{G_{\omega}}
\label{eq8}
\end{equation}
where $\theta, \omega$ are the whole framework parameters and $\beta$ is a hyperparameter adjusting the weight of VAE loss. 
The target is estimated via Double DQN~\cite{van2016deep} as $y'=r+\gamma \overline{Q}_{tot}(\boldsymbol{\tau}', \arg \max_{\boldsymbol{u}'\in \mathcal{U}^n}{Q}_{tot}(\mathbf{\boldsymbol{\tau}'}, \boldsymbol{u}'))$.
We summarize the pseudo-code of the proposed approach in Appendix~\ref{pseudo}.

\section{Experiments}
In this section, we demonstrate our experimental results of N$\text{A}^\text{2}$Q on challenging tasks over LBF~\cite{christianos2020shared} and SMAC~\cite{samvelyan2019starcraft} benchmarks.
The baselines that we select for comparison are nine popular value-based baselines, including VDN~\cite{sunehag2017value}, QMIX~\cite{rashid2018qmix}, QTRAN~\cite{son2019qtran}, Qatten~\cite{yang2020qatten}, QPLEX~\cite{wang2020qplex}, Weighted QMIX~\cite{rashid2020weighted}, DVD~\cite{li2022deconfounded}, CDS~\cite{li2021celebrating}, and SHAQ~\cite{wang2021shaq}.
The implementation details of all algorithms are provided in Appendix~\ref{app4444}, along with the benchmarks. 
All graphs showing performance results for our method, baselines, and ablations study are plotted using $\text{mean}\pm \text{std}$ with five random seeds. 
Further, we present the interpretability of N$\text{A}^\text{2}$Q to render empirical evidence about which observations are of interest to the agents, as well as the contributions of each agent and coalition. 
The source code is available at \url{https://github.com/zichuan-liu/NA2Q}.

\subsection{Level Based Foraging}
We first run the experiments on two constructed LBF tasks, wherein agents navigate a $10\times10$ grid world and collect food by cooperating with other agents if needed. 
Each agent can observe a $5\times5$ sub-grid centering around it.
When they cooperate to eat food that is smaller than their level at each step, they will receive a positive reward, otherwise, they will receive a negative reward of $-0.002$.
The action space for each agent consists of movement in four directions, eating food, and a ``none" action. 
We evaluate the performance of various algorithms with two quantities of agents and food.

\begin{figure}[h]
	%\vskip 0.2in
	\begin{center}
		\centerline{\includegraphics[width=\columnwidth]{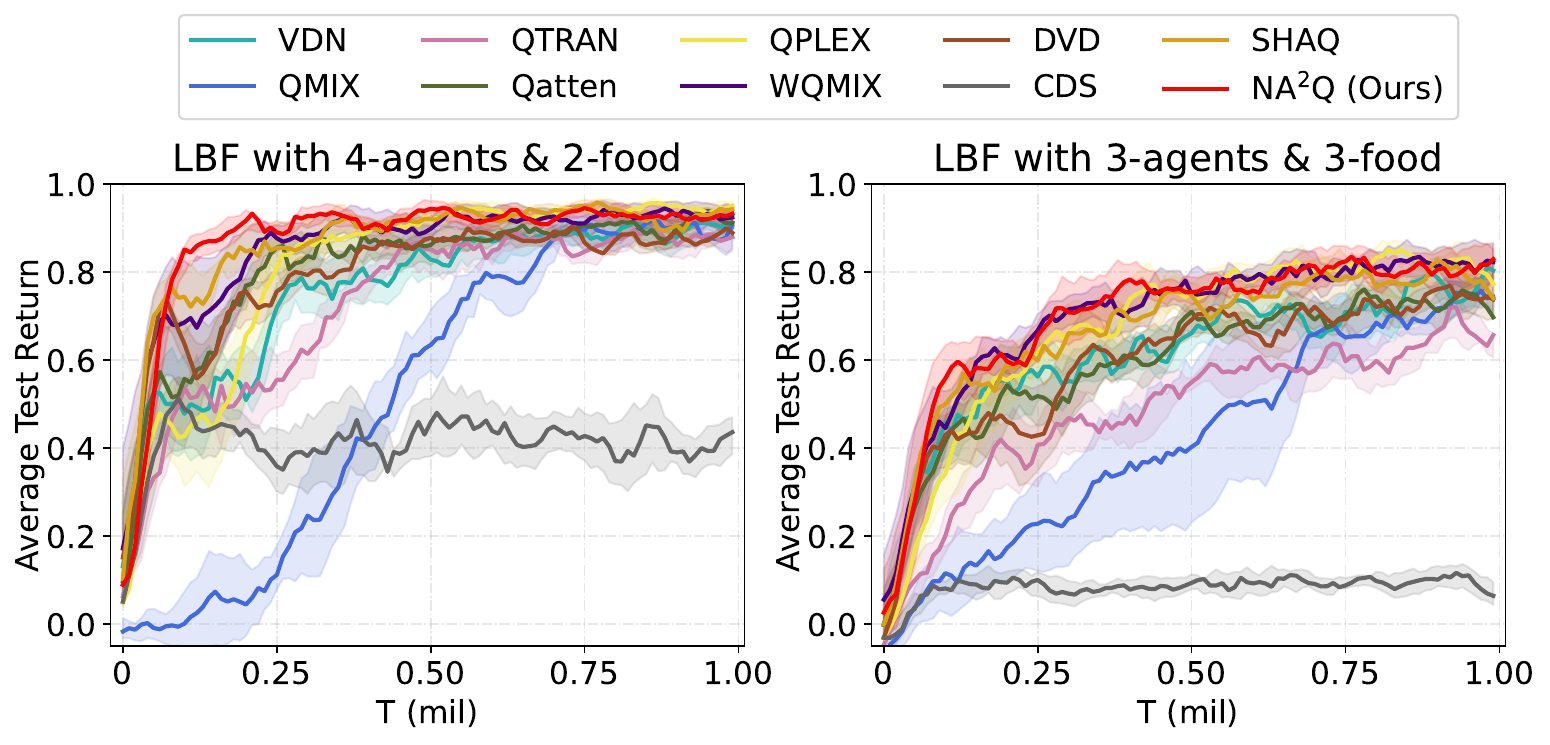}}
		\vskip -0.05in
		\caption{Average test return on two constructed tasks of LBF.}
		\label{overview_results_lbf}
	\end{center}
	\vskip -0.2in
\end{figure}

\begin{figure*}[ht]
	%	\vskip 0.2in
	\begin{center}
		\centerline{\includegraphics[width=1\linewidth]{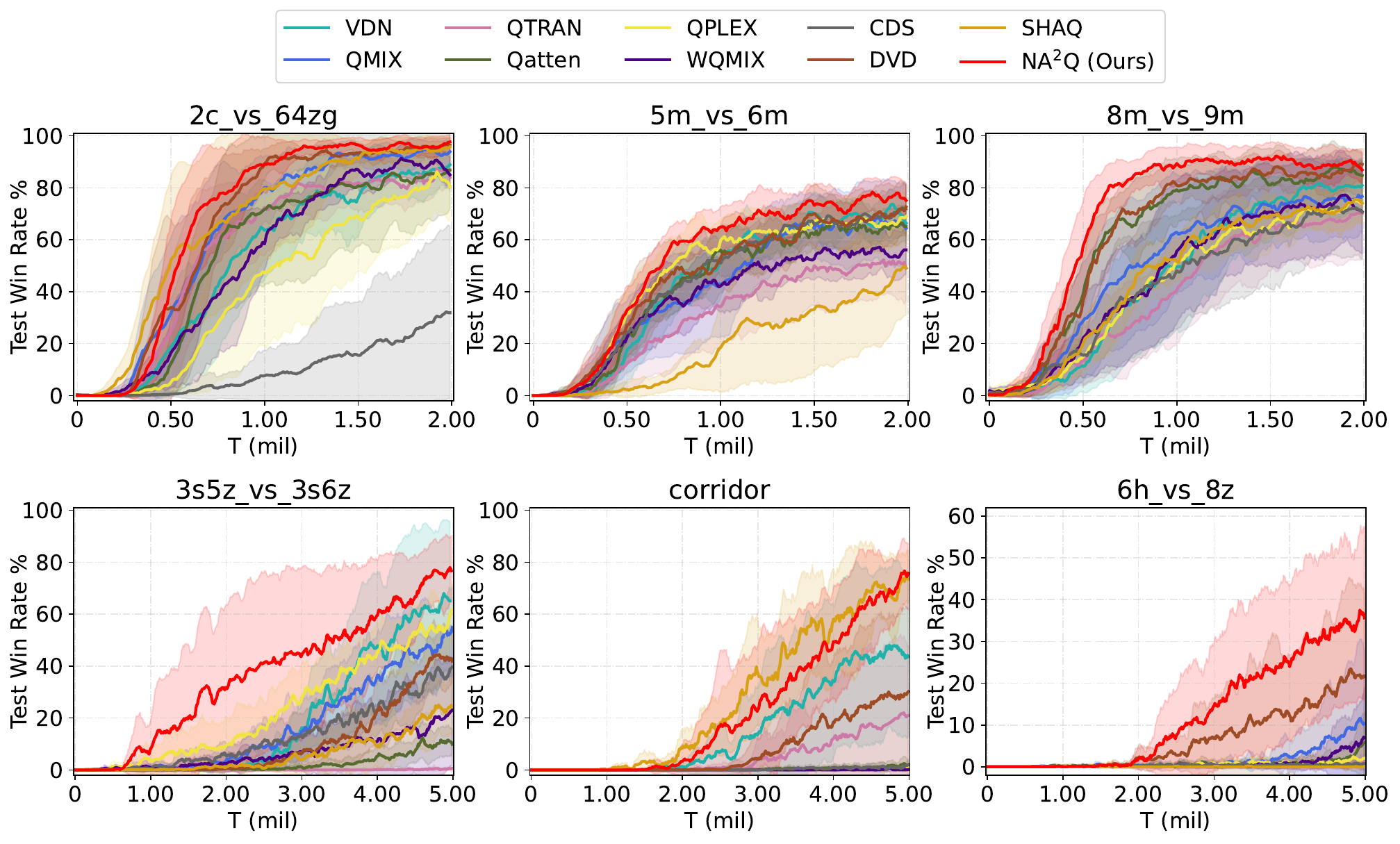}}
		\vskip -0.05in
		\caption{Test win rate \% on hard (first row), and super hard (second row) maps of SMAC benchmark.}
		\label{overview_results_main}
	\end{center}
	\vskip -0.20in
\end{figure*}

\begin{figure}[ht]
	\begin{center}
		\centerline{\includegraphics[width=\columnwidth]{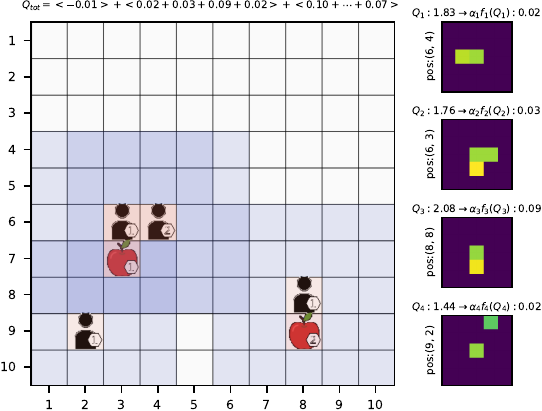}}
		\vskip -0.05in
		\caption{Visualization of the agent's mask at step~4, and the title indicates the corresponding credit assignment. The highlighted areas are the important regions for making decisions.}
		\label{inp}
	\end{center}
	\vskip -0.3in
\end{figure}

\begin{figure*}[ht]
	\centering
	\begin{minipage}{.67\textwidth}
		\centering
		\vskip -0.03in
		\subfigure[N$\text{A}^\text{2}$Q: $\varepsilon$-greedy]{
			\includegraphics[width=.32\linewidth]{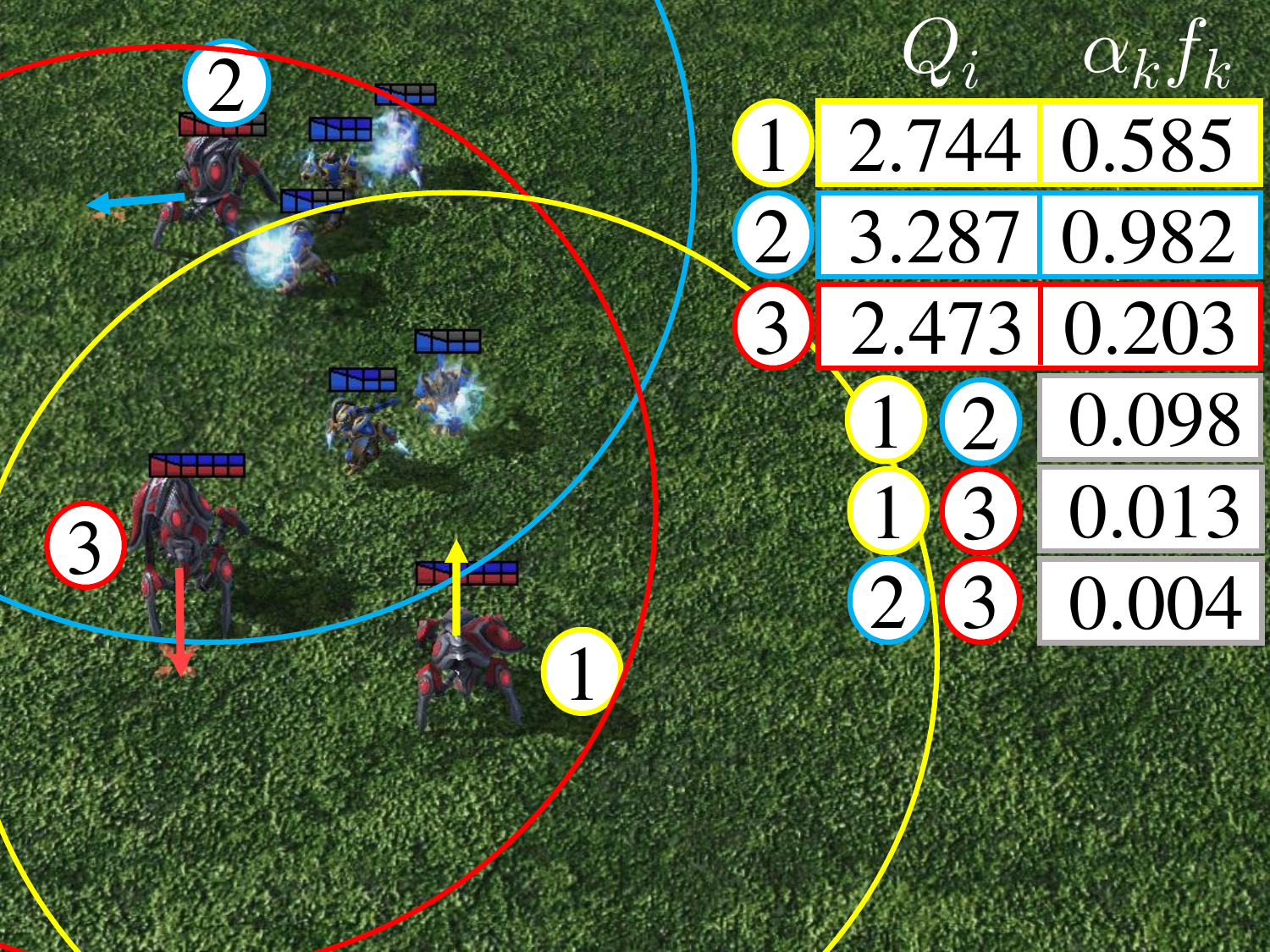}}
		\subfigure[VDN: $\varepsilon$-greedy]{
			\includegraphics[width=.32\linewidth]{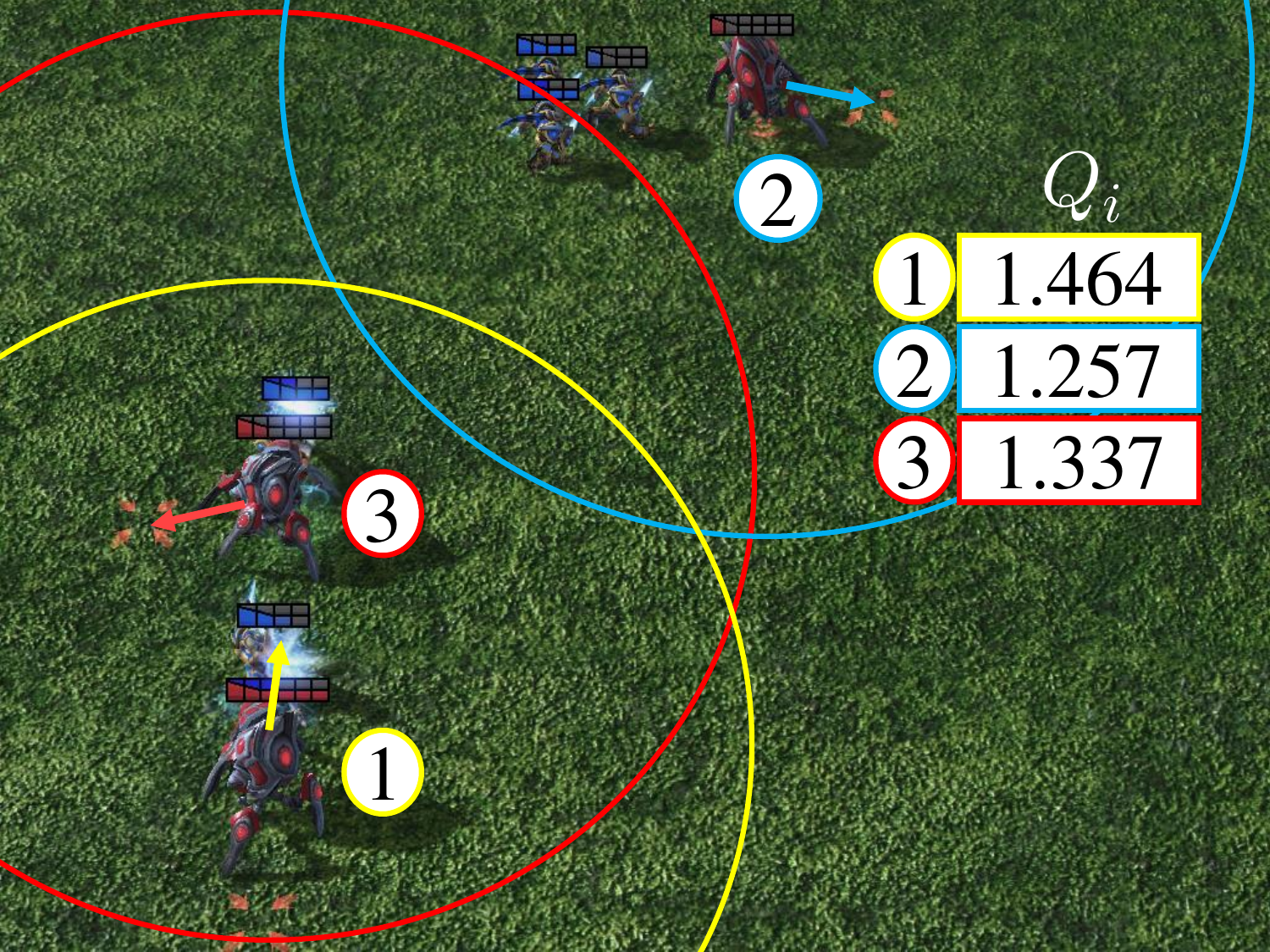}}
		\subfigure[QMIX: $\varepsilon$-greedy]{
			\includegraphics[width=.32\linewidth]{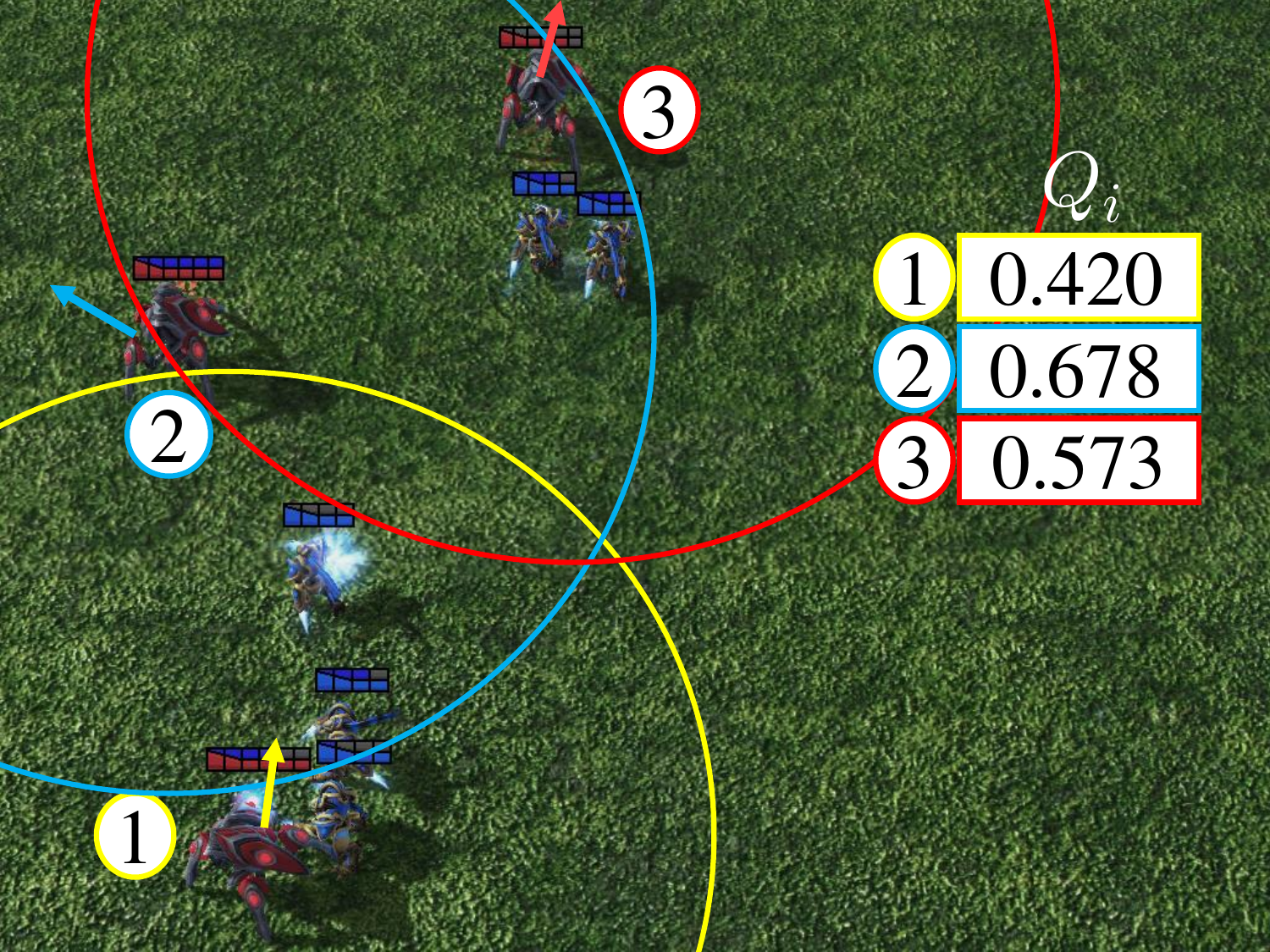}}\\
		\subfigure[N$\text{A}^\text{2}$Q: greedy]{
			\includegraphics[width=.32\linewidth]{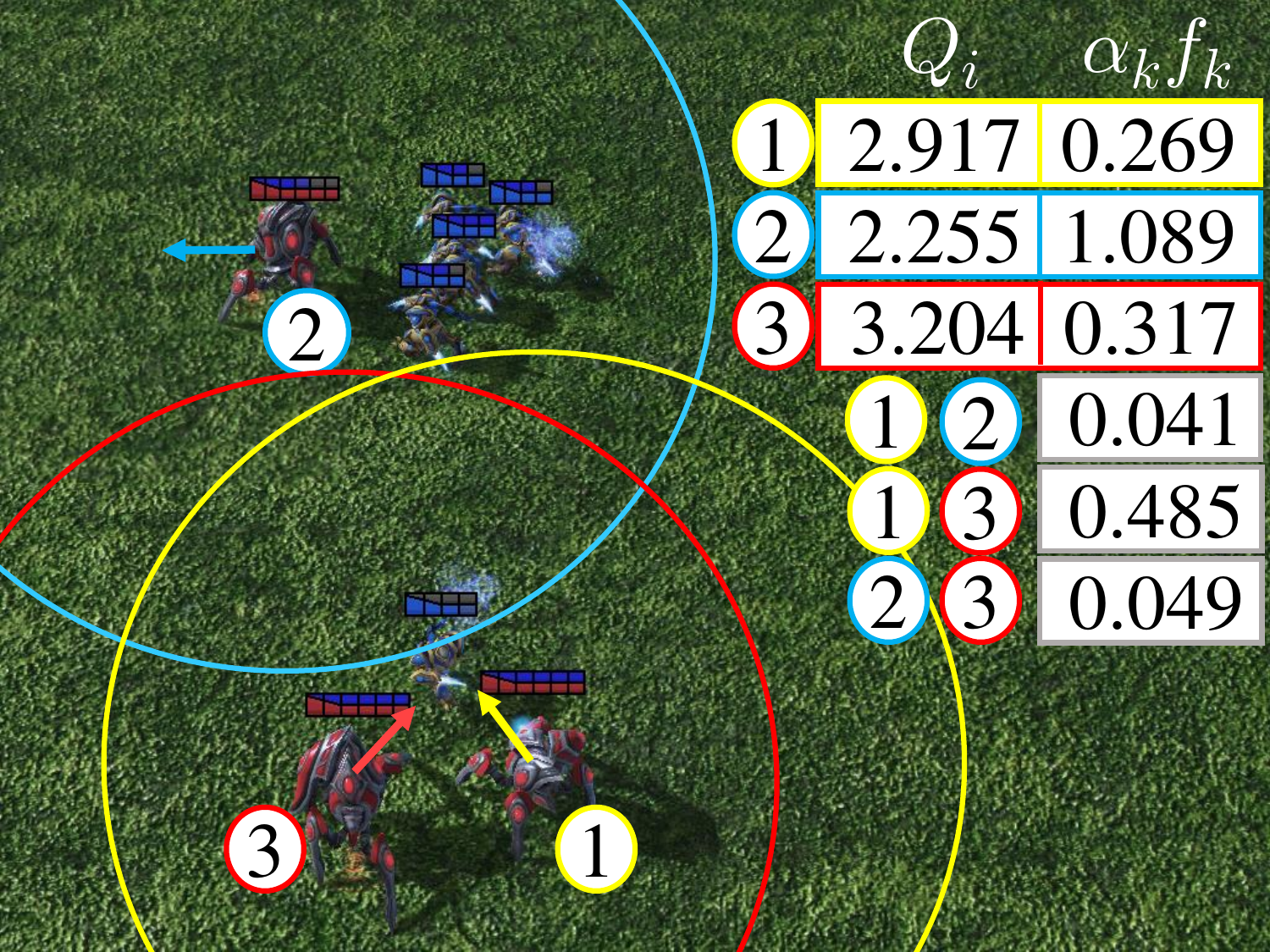}}
		\subfigure[VDN: greedy]{
			\includegraphics[width=.32\linewidth]{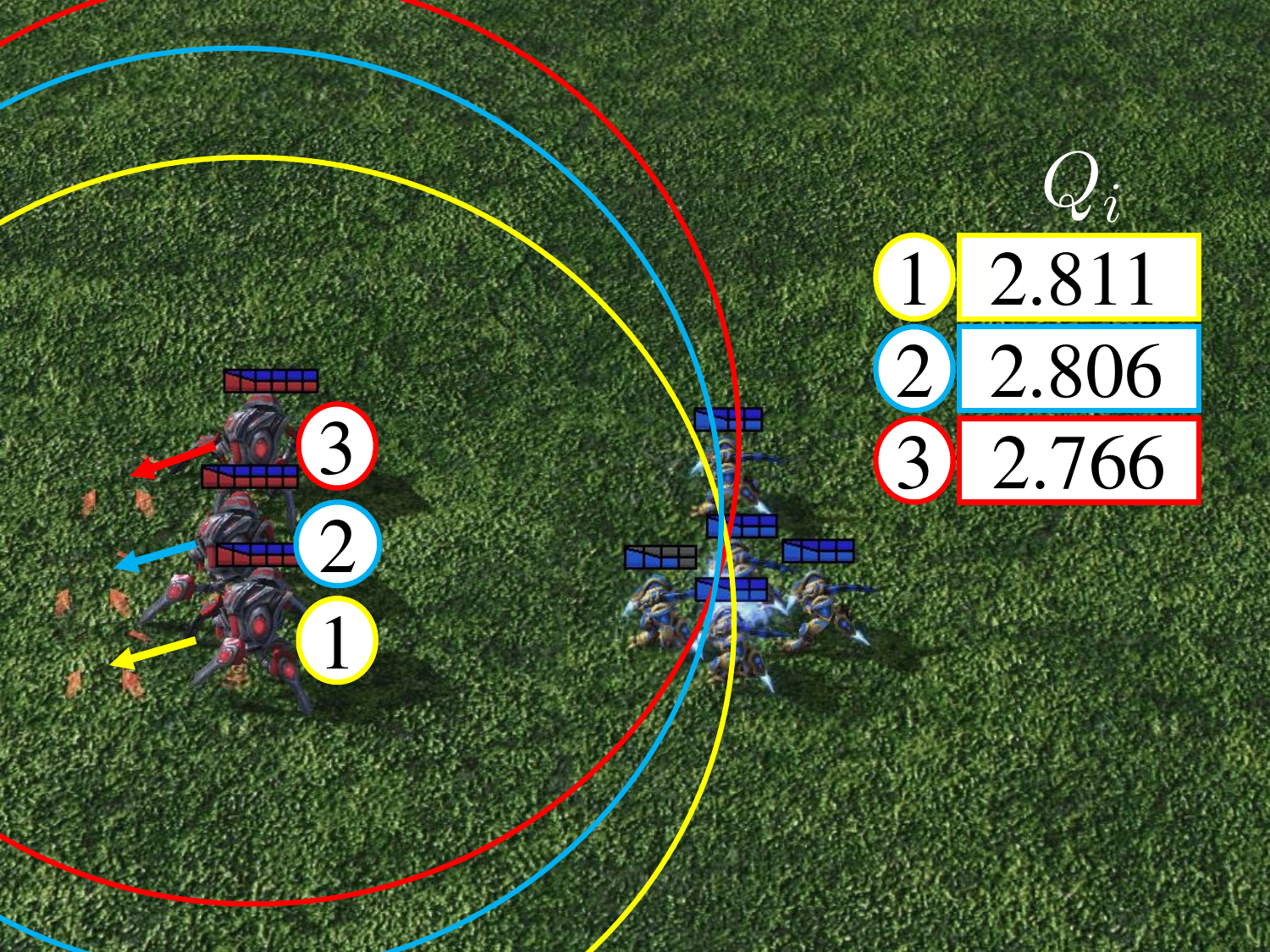}}
		\subfigure[QMIX: greedy]{
			\includegraphics[width=.32\linewidth]{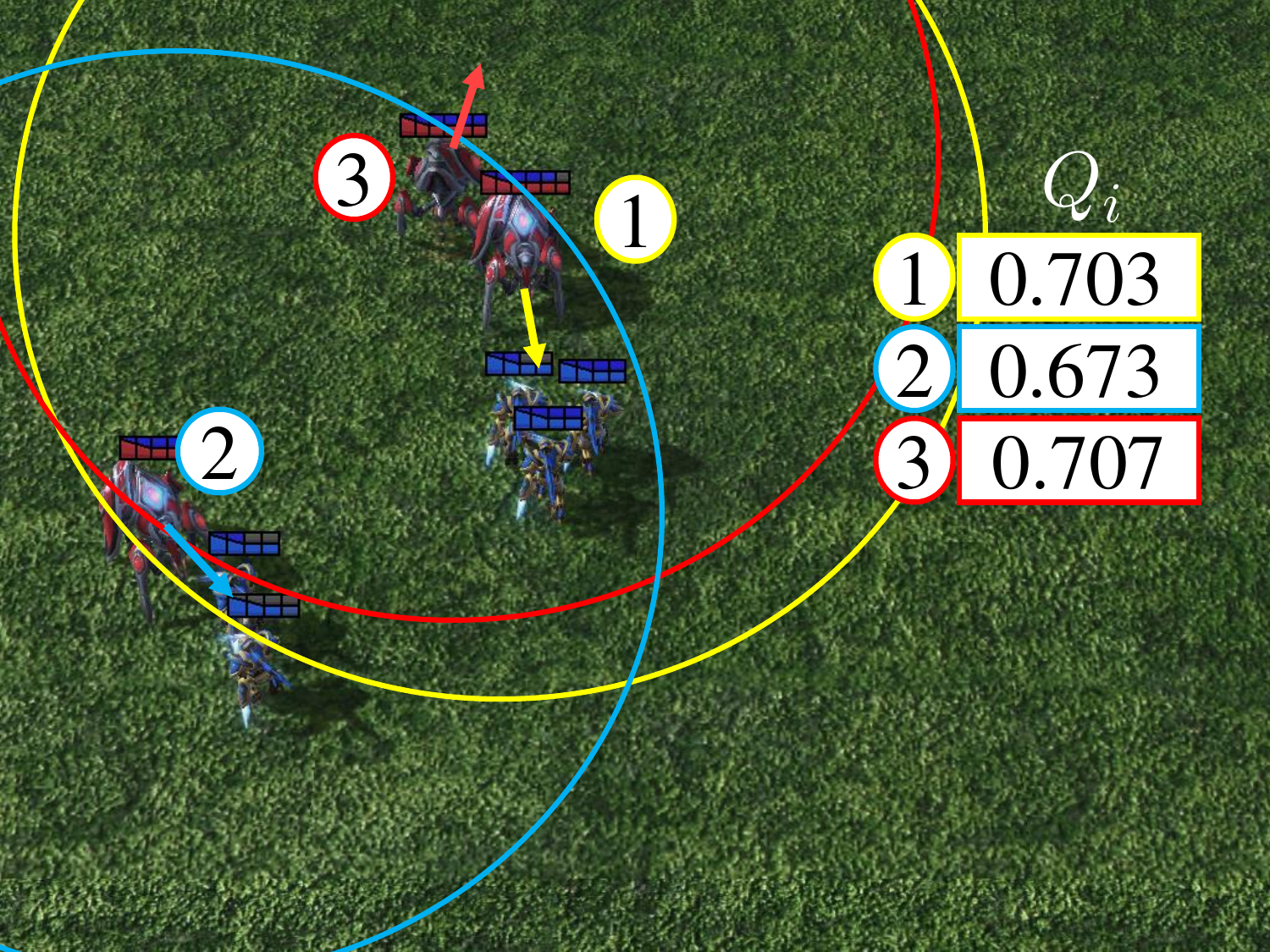}}
		\vskip -0.1in
		\caption{Visualization of evaluation for N$\text{A}^\text{2}$Q and baselines on 3s\_vs\_5z map. 
			Different colored circles indicate the corresponding central attack range, while arrows indicate movement or attack direction. 
			Each decomposed Q-value is displayed at the top-right, and for N$\text{A}^\text{2}$Q, we report the contribution of its unary and pairwise shape functions to the team.
			This shows that the values of VDN and QMIX are difficult to explain, while the Q-values of the decomposition by N$\text{A}^\text{2}$Q intend to correspond more clearly to the actions.
		}
		\label{fig2}
	\end{minipage}
	\hspace{.1in}
	\begin{minipage}{.3\textwidth}
		\centering
		\includegraphics[width=\linewidth]{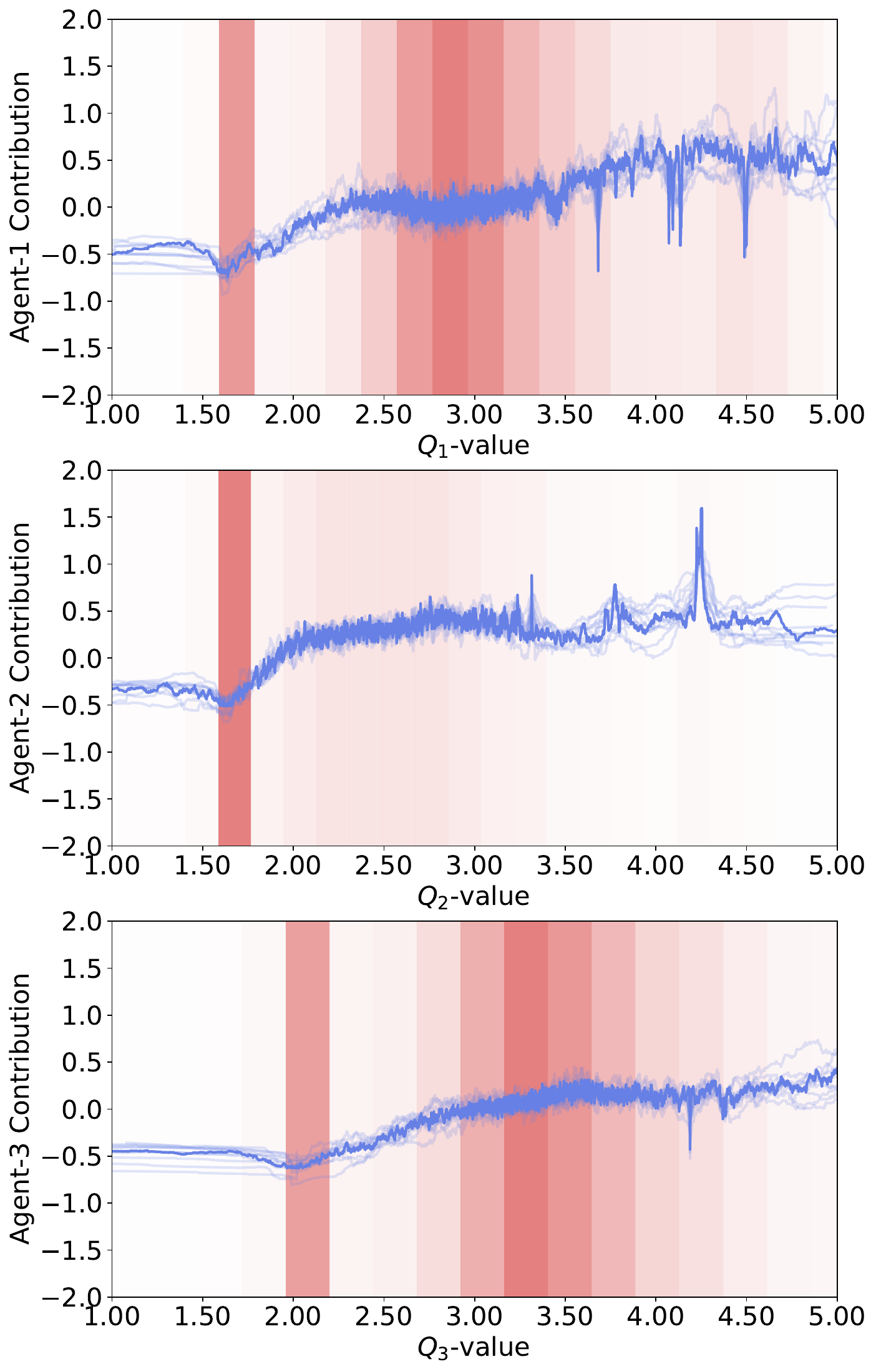}
		\vskip -0.2in
		\caption{Learned shape function $f_i$ by trained N$\text{A}^\text{2}$Q on 3s\_vs\_5z scenario. As expected, individual Q-values increase with the contribution of the agent.}
		\label{fig1}
	\end{minipage}
	\vskip -0.1in
\end{figure*}

\textbf{Performance on LBF.} 
\textit{Figure~\ref{overview_results_lbf}} shows the performance comparison against baselines on two constructed tasks of LBF. 
Our method achieves competitive performance in LBF tasks, demonstrating its efficiency across a range of scenarios. 
The failure of CDS may be due to the inability of diverse agents to explore collaborative strategies.
VDN, QMIX, and QTRAN require more steps to discover sophisticated policies, indicating that they are in trouble due to the limitations of representing spurious relationships between credits and decomposed Q-values.
QPLEX receives a lower reward than N$\text{A}^\text{2}$Q before $0.5M$ timesteps, which may need more time steps to explore since the complex network architecture.
Compared to QMIX, DVD obtains improved performance since it utilizes the de-confounded training mechanism.
N$\text{A}^\text{2}$Q achieves slightly higher performance than Qatten.
It implies that considering higher-order interactions and fine-grained learning semantics can promote credit assignment and correctly guide decentralized agents.
Compared to WQMIX and SHAQ, N$\text{A}^\text{2}$Q achieves nearly the same performance with better robustness. 
The reason may be that providing a reasonable inference path for credit assignment can assist in improving coordination.

%\textit{Figure~\ref{overview_results_lbf}} shows the performance comparison against baselines on two constructed tasks of LBF. 
%%Our method consistently outperformsall the tasks, manifesting that N$\text{A}^\text{2}$Q is not restricted by these scenarios.
%Our method consistently outperforms all other methods, demonstrating its flexibility and effectiveness across a range of scenarios. 
%While the failure of CDS may be due to the inability of diverse agents in exploring collaborative strategies.
%VDN, QMIX, and QTRAN require more steps to discover sophisticated policies, indicating that they are in trouble due to the limitations of representing spurious relationships between credits and decomposed Q-values.
%QPLEX has a lower reward than N$\text{A}^\text{2}$Q before $0.5M$ timesteps, which may need more time steps to explore since the complex network architecture.
%Compared to Weighted QMIX and SHAQ, N$\text{A}^\text{2}$Q achieves nearly the same performance with better robustness. 
%The reason may be that bringing a reasonable inference path for credit assignment can assist in improving coordination.

\textbf{Interpretability of N$\text{A}^\text{2}$Q.}
To verify that N$\text{A}^\text{2}$Q possesses interpretability, we show its behavior matches that of corresponding agents on LBF.  
\textit{Figure~\ref{inp}} illustrates the small regions that each agent focuses on, and the headings labeled to show its credits.
It is evident that each agent captures task-relevant semantic information (the highlighted areas in the heatmaps) to make decisions. 
Specifically, agents pay more attention to the food position within their sight range, and tend to cooperate with teammates when the level of food is higher than themselves. 
Agents 1 and 2 only obtain credit assignments with $0.02$ and $0.03$, respectively, however, their pairwise shape function $f_{12}$ obtains high credit with $0.10$.
This implies that they have captured the cooperative skill for eating the food, which should be attributed to considering the different orders of the coalition of agents in designing the value decomposition mechanism.
Indeed, this semantic interpretation is also consistent with human visual patterns~\cite{greydanus2018visualizing} that tend to focus selectively on parts of the visual space and form collaborative relationships. 
We likewise show the interpretation for the whole episode, which can be found in Appendix~\ref{LBF}.

\subsection{StarCraft Multi-Agent Challenge}
Further, to broadly compare the performance of N$\text{A}^\text{2}$Q with baselines, we conduct experiments on the more challenging SMAC benchmark, which is a commonly used testbed for MARL algorithms. 
At each timestamp, each agent receives local observations and then obtains a global reward after making a move or attacking its enemies.
We compare the performance of N$\text{A}^\text{2}$Q with other baselines on $12$ different scenarios, including easy, hard, and super hard scenarios.
The details of these scenarios can be found in Appendix~\ref{4fafas}.

\textbf{Performance on SMAC.} 
The experimental results for different scenarios are shown in \textit{Figure~\ref{overview_results_main}}.
We can find that N$\text{A}^\text{2}$Q could consistently gain almost the best performance on all scenarios, especially on the super hard tasks.
QTRAN does not yield satisfactory performance, which may be due to the relaxation in practice that is insufficient for challenging domains.
Both baseline VDN and QMIX can achieve satisfactory performance on some easy or hard maps, such as 5m\_vs\_6m, but in super hard maps they fail to well solve the tasks.
Intuitively, super hard scenarios require more coordination skills, while their mixing network hardly captures the different interaction relationships among agents.
Similarly, QPLEX and WQMIX do not perform well despite relaxed restrictions on the joint value function, which may contribute to inefficient value decomposition without considering the local semantics.
Qatten falls short in satisfactory performance on super hard tasks, which implies that the lack of finely learned individual semantics brings about a spurious correlation between $s$ and $Q_{tot}$ and thus limits performance.
One possible reason for CDS not performing as well as reported by~\citet{li2021celebrating} is that paying more attention to policy diversity leads to instability during the learning stage, especially in less-agent maps, e.g., 2c\_vs\_64zg.
SHAQ only achieves comparable performance with N$\text{A}^\text{2}$Q in the corridor map, which seems to have difficulty adapting to all scenarios.
The reason could be that SHAQ ignores high-order interactions among agents.
DVD only attains comparable performance with N$\text{A}^\text{2}$Q on 2c\_vs\_64zg map, and struggles to achieve competitive performance on the other scenarios, probably due to the fact that it neglects to explicitly consider high-order interactions among agents.
In particular, for super hard task 6h\_vs\_8z, N$\text{A}^\text{2}$Q still maintains superior performance, while almost all the baselines are unable to learn efficient strategies.
It validates that enriching shape functions for estimating credits over each agent and the coalition of agents can boost efficient value decomposition.
In summary, our approach achieves impressive performance on all scenarios, showing the advantage of N$\text{A}^\text{2}$Q with attentive design.
More empirical results can be referred to \textit{Figure~\ref{overview_results_app}} in Appendix~\ref{Extra}.

\begin{figure*}[ht]
	\centering %图片全局居中
	%	\subfigbottomskip=0pt %两行子图之间的行间距
	\subfigcapskip=-5pt %设置子图与子标题之间的距离
	\subfigure[Number of interactive order terms]{
		\includegraphics[width=\linewidth]{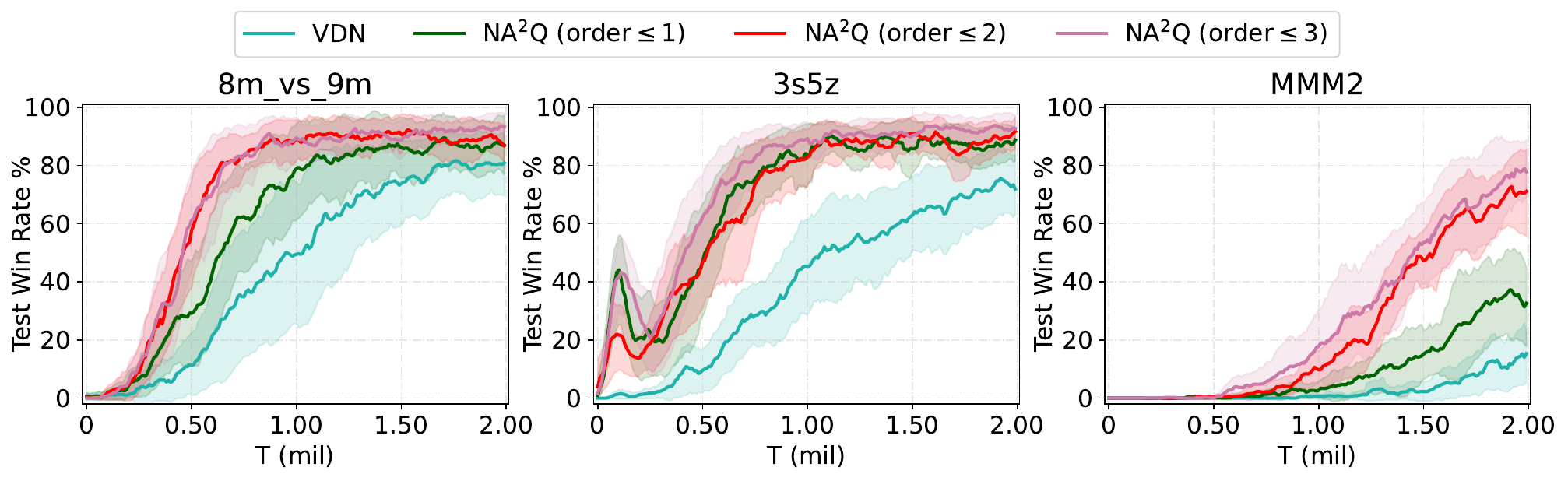}}\label{ablation1}
	\vskip -0.03in
	\subfigure[Influence of identity semantics and attention]{
		\includegraphics[width=\linewidth]{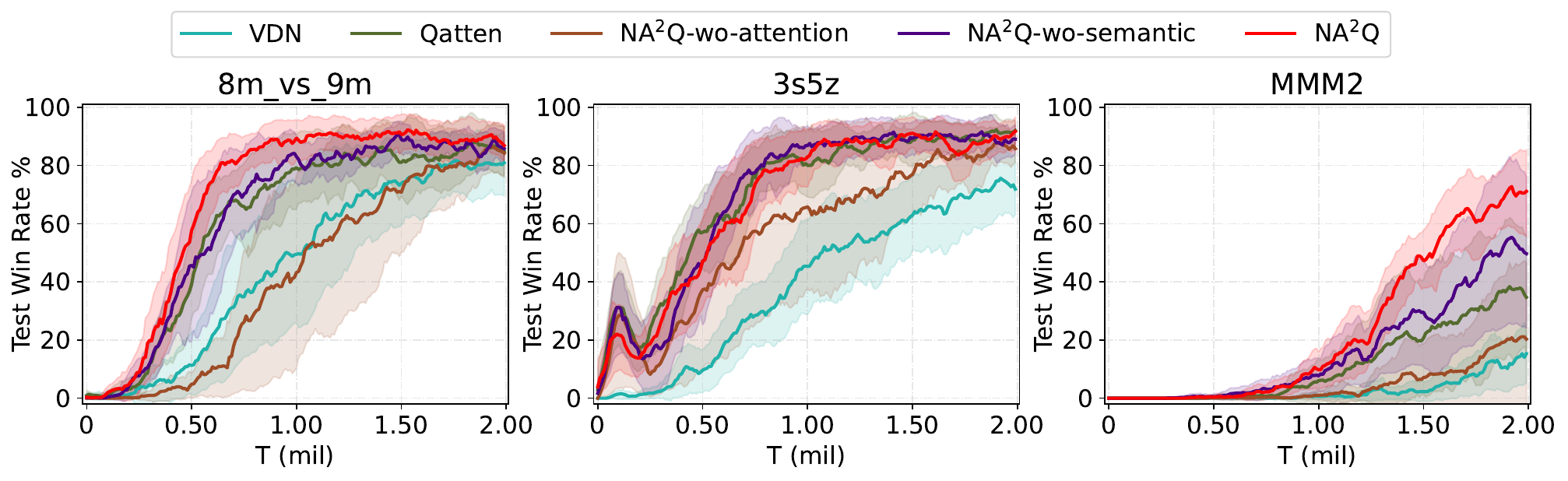}}\label{ablation2}
	\vskip -0.13in
	\caption{Ablation studies of N$\text{A}^\text{2}$Q on SMAC benchmark.}
	\vskip -0.15in
	\label{ablation}
\end{figure*}

\textbf{Interpretability and Stability.}
To intuitively show the interpretability of N$\text{A}^\text{2}$Q, we display some keyframes on 3s\_vs\_5z scenario as shown in \textit{Figure~\ref{fig2}}. 
We first consider the suboptimal action by $\varepsilon$-greedy. 
As seen from \textit{Figure~\ref{fig2}(a)}, Agent-3 escapes from its teammates and receives a lower contribution relative to the allies, which can be understood as meaning that it does not contribute to the team.
Meanwhile, N$\text{A}^\text{2}$Q can also provide pairwise contributions among agents, whose contributions are close to $0$ when they do not collaborate. 
However, it is hard to explicitly understand the behavior of VDN and QMIX from their Q-values.
As for optimal actions of N$\text{A}^\text{2}$Q, Agent-2 plays an essential role in kiting enemies at this time step and obtains a higher contribution of $1.089$, which is a crucial trick to victory that the agent can learn how to kite the enemies effectively~\cite{samvelyan2019starcraft}.
An interesting finding is that Agent-1 and Agent-3 siege the enemy and their coalition contribution is remarkably higher than other coalitions.
This shows the advantage of considering different orders of interactions among agents, which can facilitate deducing the contributions of each agent in value decomposition.
Whereas VDN produces the same action and does not possess an explicit interpretation as N$\text{A}^\text{2}$Q since it only considers order-$1$ for credit assignment.
For QMIX, the agents' behaviors are also difficult to understand because the Q-values are roughly equal.
A detailed description of N$\text{A}^\text{2}$Q about observation semantics and agent contributions is deferred to Appendix~\ref{maskSMAC}.

The interpretability of models is intrinsically coupled to their stability.
To assess stability, we evaluate $10$ models with different random seeds on 3s\_vs\_5z scenario, collecting $32$ rounds of interaction data for each model and plotting the shape functions with semi-transparent lines.
As shown in~\textit{Figure~\ref{fig1}}, we display the average contribution of each shape function, where the blue lines and pink bars indicate the contribution of each agent to the team and the Q-value distribution, respectively, where a bar with an intense color means larger samples located here.
As expected, the crimson area on the leftmost side represents its Q-value when an agent died, which means the agent had a lower contribution to the team.
Most samples of Agent-1 and Agent-3 are gathered around $3.00$ with larger positive contribution values, which implies that they spend more time step attacking enemies.
We also find that most samples of Agent-2 are gathered with a negative contribution value.
The reason may be that Agent-2 pays more attention to kiting the enemies, causing deaths to always occur earlier and having more dead Q-values.
Moreover, we compute the standard deviation of the plotted shape functions to be $0.124$, and the shape functions do not appear to deviate significantly, even for a few data points (white/light areas).
This finding attests to the robustness and resilience of our novel value decomposition mechanism, further enhancing interpretability.

\subsection{Ablation Study}
To understand the impact of each component in the proposed method, we conduct ablation studies to answer the following questions: (a) How does the model's performance benefit from the number of interaction orders among agents? (b) How do identity semantics influence performance? (c) Whether the intervention function is rational for value decomposition?
To study component (a), we ablate shape functions for different order numbers named \textit{N$\text{A}^\text{2}$Q (order$\leq l$)} in Eq.~(\ref{sdada}).
Since higher-order interactions will decrease computational efficiency due to permutations, we select three different order interactions by setting $1\leq l\leq3$. 
To study components (b) and (c), \textit{N$\text{A}^\text{2}$Q-w/o-semantic} represents replacing identity semantics $\boldsymbol{z}$ with global state $\boldsymbol{s}$ in Eq.~(\ref{eq4}), and \textit{N$\text{A}^\text{2}$Q-w/o-attention} represents ablating attention mechanism for credit assignment, respectively.
Additionally, since VDN and Qatten can be seen as special order-$1$ shape functions, we take them as a baseline for comparison.

We carry out ablation studies on three hard and super hard scenarios, and present the results in \textit{Figure~\ref{ablation}}.
As shown in \textit{Figure~\ref{ablation}(a)}, N$\text{A}^\text{2}$Q achieves better performance as the number of order interactions increases, which validates the importance of considering higher-order interaction relationships among agents.
Why not then have the number of order interactions as large as possible? 
A potential drawback is that an excessive number of order interactions might hurt interpretability, as shape functions beyond pairs are harder to visualize.
Generally, moderate order terms (e.g., $l\leq 2$) are enough for an appropriate trade-off between performance improvement and interpretability. 
In \textit{Figure~\ref{ablation}(b)}, the ablation of each part of our intervention function brings a noticeable decrease in performance.
Specifically, the performance of \textit{N$\text{A}^\text{2}$Q-w/o-attention} and VDN decreases, which indicates that the global state information is beneficial to estimate the credit assignment. 
Besides, the performance of \textit{N$\text{A}^\text{2}$Q-w/o-attention} is slightly higher than VDN because it considers more possible interactions among agents, leading to more capabilities than linear order-1 interactions.
\textit{N$\text{A}^\text{2}$Q-w/o-semantic} performs slightly worse than N$\text{A}^\text{2}$Q, which indicates the fine learning identity semantics own the greater representation ability to keep track of the feature influence of each agent.
Additionally, the performance of \textit{N$\text{A}^\text{2}$Q-w/o-semantic} is consistently superior to Qatten on a range of tasks, which implies that high-order interactions among agents can provide more capacity to search for efficient patterns of cooperation.
% \textit{N$\text{A}^\text{2}$Q-w/o-attention} performs similarly to VDN with large variance, suggesting that consistent credit could limit agents to find efficient cooperation patterns. 
To summarize, N$\text{A}^\text{2}$Q that is conditioned on all parts gives the best performance while retaining interpretability, which improves flexibility and saves human labor.

%We carry out ablation studies on 8m\_vs\_9m and MMM2 scenarios, and present the results in \textit{Figure~\ref{ablation}}.
%As shown in \textit{Figure~\ref{ablation}(a)}, N$\text{A}^\text{2}$Q achieves better performance as the number of order interactions increases, which validates the importance of considering higher-order interaction relationships among agents.
%Why not then have the number of order interactions as large as possible? 
%A potential drawback is that a too large number of order interactions might hurt interpretability since beyond pairwise shape functions are harder to visualize.
%Generally, moderate order terms (e.g., $l\leq 2$) are enough for an appropriate trade-off between performance improvement and interpretability. 
%%This means that the ``sweet spot" seems to rest on pairwise interactions.
%In \textit{Figure~\ref{ablation}(b)}, the ablation of each part of our intervention function brings a noticeable decrease in performance.
%\textit{N$\text{A}^\text{2}$Q-w/o-attention} performs similarly to VDN with large variance, suggesting that consistent credit could limit agents to find efficient cooperation patterns. 
%\textit{N$\text{A}^\text{2}$Q-w/o-semantic} performs even worse than N$\text{A}^\text{2}$Q, which indicates the importance of identity semantics for estimating credits.
%To summarize, N$\text{A}^\text{2}$Q that is conditioned on all parts gives the best performance while retaining interpretability, which improves flexibility and saves human labor.

\section{Conclusion}
In this paper, we present N$\text{A}^\text{2}$Q in the scope of value decomposition, which combines the inherent interpretability of GAMs, opening the door for other advances in the interpretability perspective of MARL.
N$\text{A}^\text{2}$Q allows for end-to-end training in a centralized fashion and models higher-order interactions to deduce precise credit for executing decentralized policies.
%To that end, we propose the Neural Attention Additive Q-learning~(N$\text{A}^\text{2}$Q), a novel mixing architecture to understand agents' own contributions and their interaction credits.
%We give a rigorous perspective that N$\text{A}^\text{2}$Q provides precise guarantees on the approximate Q-learning generalization.
Moreover, we provide local semantic masks as evidence for decision-making.
The empirical results show that N$\text{A}^\text{2}$Q enjoys its interpretability and scalability while maintaining competitive performance.
We believe that our work proves a solid basis for further research and could catalyze the community's effort toward understanding cooperative tasks. 
A promising direction for future work is improving the performance of N$\text{A}^\text{2}$Q by considering higher-order coalitions of agents.
However, they might worsen the intelligibility of the learned N$\text{A}^\text{2}$Q with higher-order agent interactions, especially as the number of agents increases.
It would be interesting to explore an efficient representation for interpreting a large-scale agent system, such as clustering similar terms in the N$\text{A}^\text{2}$Q framework.

\section{Acknowledgements}
The work was supported by the National Natural Science Foundation of China under Grant 62073160.

\bibliography{example_paper}
\bibliographystyle{icml2023}

%%%%%%%%%%%%%%%%%%%%%%%%%%%%%%%%%%%%%%%%%%%%%%%%%%%%%%%%%%%%%%%%%%%%%%%%%%%%%%%
%%%%%%%%%%%%%%%%%%%%%%%%%%%%%%%%%%%%%%%%%%%%%%%%%%%%%%%%%%%%%%%%%%%%%%%%%%%%%%%
% APPENDIX
%%%%%%%%%%%%%%%%%%%%%%%%%%%%%%%%%%%%%%%%%%%%%%%%%%%%%%%%%%%%%%%%%%%%%%%%%%%%%%%
%%%%%%%%%%%%%%%%%%%%%%%%%%%%%%%%%%%%%%%%%%%%%%%%%%%%%%%%%%%%%%%%%%%%%%%%%%%%%%%

\newpage
\appendix
\onecolumn

\section{Credit Assignment for Value Decomposition Algorithms}\label{appendix1}
Previous work~\cite{li2022deconfounded} defined the general formula for credit assignment in value decomposition methods as
\begin{equation}
Q_{tot} = \sum_{k=1}^m \alpha_{k} \widehat{Q}_k,
\label{adsasdasd}
\end{equation}
where $\widehat{Q}_k$ is transformed as a temporal value by $f_k(\cdot)$ and $\alpha_{k}$ denotes a credit that expresses the contribution of the temporal value to the joint action-value $Q_{tot}$. 
This formula can be applied for generalization in widely investigated approaches of mixing networks based on value decomposition, next we introduce these methods in detail\footnote{For convenience, all bias networks are omitted if existing.}. 

The first method is VDN~\cite{sunehag2017value}, which seeks to learn a joint value function $Q_{tot}(\boldsymbol{\tau}, \boldsymbol{u})$ via equal credit assignment. It represents $Q_{tot}$ as the sum of all individual value functions as $Q_{tot} = \sum_{i=1}^n Q_i$ without the use of additional state information., where Eq.~(\ref{adsasdasd}) can be rewritten when $m = n$, $\alpha_{k}=1$, and $\widehat{Q}_i=f_i(Q_i)$. 

More common algorithms transform the local Q-values into the temporal Q-values via the global state $\boldsymbol{s}$. For example, QMIX~\cite{rashid2018qmix} can be represented by a monotonic neural network $f_{1 \ldots n}(\cdot)$ with the global state $\boldsymbol{s}$ as
\begin{equation*}
[\widehat{Q}_k]_{k=1}^m = f_{1 \ldots n}(Q_1, \ldots, Q_n), \frac{\partial f_{1 \ldots n}}{\partial Q_i}>0,
\end{equation*}
where $k\in\{1,\cdots,m\}$ denotes the embedding number. Then the credit $\alpha_k(\boldsymbol{s})$ is calculated by another monotonic neural network and utilized in Eq.~(\ref{adsasdasd}). Some methods that improve on QMIX, e.g, Qatten~\cite{yang2020qatten} replace the neural network $f_{1 \ldots n}$ into an attention mechanism, Weighted QMIX~\cite{rashid2020weighted} uses different weights on TD error, and CDS~\cite{li2021celebrating} improves diversity among agents by constructing intrinsic rewards.

Further, QPLEX~\cite{wang2020qplex} combines QMIX and VDN in a dueling mixing network as
\begin{equation*}
Q_{tot} = \sum_i^n \alpha_{i}\widehat{Q}_i + \alpha_{1 \ldots n}f_{1 \ldots n}(Q_1, \ldots, Q_n), \frac{\partial f_{1 \ldots n}}{\partial Q_i}>0,
\end{equation*}
where $\widehat{Q}_i=Q_i$ represents the local temporal value and $f_{1 \ldots n}$ represents the advantage function to get $\widehat{Q}_{1 \ldots n}$, which also uses an attention mechanism. Therefore, it is equivalent to Eq.~(\ref{adsasdasd}) when $m=n+1$ and $\alpha_k\in\{\alpha_1,\cdots, \alpha_n, \alpha_{1 \ldots n}\}^m$.
It is straightforward to notice that QPLEX is the sum of term $\normalsize{\textcircled{\scriptsize{1}}}\normalsize$ and term $\normalsize{\textcircled{\scriptsize{2}}}\normalsize$ in Eq.~(\ref{sdada}).

The last method SHAQ~\cite{wang2021shaq} improves the credit assignment of QMIX via Shapley theory for interpretation, which can also be expressed by Eq.~(\ref{adsasdasd}).

\section{Approximation Guarantees for N$\text{A}^\text{2}$Q}\label{app2}
Inspired by non-linear GAMs, e.g., NAM~\cite{agarwal2021neural} and SPAM~\cite{dubey2022scalable}, we modify the decomposition of ${Q} = [Q_i]_{i=1}^n\in \mathcal{Q}$ in Eq.~\ref{sdada} by rewriting the order number $1\leq l\leq n$ with the shape functions as 
\begin{equation}
Q_{tot} = f_0+ \lambda _{1d} \cdot  \left<\boldsymbol{a}_{1d}, F_1({Q})\right>+\sum_{d=1}^{\rho_2}\lambda _{2d} \cdot \left<\boldsymbol{a}_{2d}, F_2({Q})\right>^2+\cdots+\sum_{d=1}^{\rho_n}\lambda _{nd} \cdot  \left<\boldsymbol{a}_{nd}, F_n({Q})\right>^n,
\label{keyeq}
\end{equation}
where $\{\lambda_{ld}\}_{d=1}^{\rho_l}$ and $\{\boldsymbol{a}_{ld}\}_{d=1}^{\rho_l}$ are the corresponding eigenvalues and bases for credit matrix $\boldsymbol{\alpha}_l = \left\{ \alpha_{\mathcal{D}_l} \right\}$ to represent the order-$l$ interactions between all non-empty subsets of $l\in \mathcal{N}$, 
$\rho_l\in \{1,\rho_2,\cdots,\rho_n\}$ denotes the
rank of the tensor, and the function $F_l({Q}) = [f_{l1}(\cdot), f_{l2}(\cdot), \cdots, f_{ln}(\cdot)]\in \mathcal{F}_l$
is a family of shape functions in the order-$l$. 
Next, we present learning-theoretic and approximation guarantees for this type of enrichment, with a more precise regret bound.

\begin{assumption}\label{ass1}
	($\eta $ - Exponential Spectral Decay of Approximation.) For the family of all  decomposition  $Q\in \mathcal{Q}$ as outlined in Eq.~(\ref{keyeq}), we assume that there exist absolute constants $C_1 < 1$ and $C_2 = \mathcal{O}(1)$ such that $\lambda_{ld} \leq C_1 \exp(-C_2\cdot d^\eta )$ for each $l\in \mathcal{N}$ and  $d\geq 1$.
\end{assumption}
Assumption~\ref{ass1} provides a soft threshold for singular value decay, i.e., implying that only a few decay degrees of freedom are sufficient to accurately approximate $f_k$. 
We consider the general results under the \textit{1-Lipschitz} loss approximated by this enrichment decomposition of the metric regret bound. 
Let us denote the Taylor expansion decomposition in Eq.~(\ref{eqqqq}) as $Q_{tot}(\boldsymbol{\tau}, \boldsymbol{u}):\mathcal{Q}\to\overline{\mathcal{Y}}$.
Thus, we aim to bound the expected risk in Eq.~(\ref{eqqqq}) with the empirical risk in Eq.~(\ref{sdada}) to demonstrate that learning an enrichment decomposition method does not incur a  larger error compared with learning the Taylor expansion.
At a high level, for any function $Q_{tot}(\boldsymbol{\tau}, \boldsymbol{u}):\mathcal{Q}\to  \mathcal{Y}$ and bounded \textit{1-Lipschitz} loss $\ell: \mathcal{Y}\times \mathcal{Y}\to [0, 1]$, the empirical risk over $b$ samples from $\mathcal{B}$ as $\widehat{\mathcal{L}}_b(Q_{tot}(\boldsymbol{\tau}, \boldsymbol{u}))=\frac{1}{b}\sum_{j=1}^b\ell(Q_{tot}, y)$. We donate $\widehat{Q}_{tot}$ as the \textit{empirical risk minimizer}, then,
\begin{equation}
\widehat{Q}_{tot} ={\arg \min}_{Q_{tot}\in \mathcal{Y}} \widehat{\mathcal{L}}_b(Q_{tot}(\boldsymbol{\tau}, \boldsymbol{u})).
\end{equation}
Similarly, the expected risk can be given, over the sample distribution $\mathfrak{P}$ as $\mathcal{L}(Q_{tot}(\boldsymbol{\tau}, \boldsymbol{u}))=\mathbb{E}_{([Q_i]_{i=1}^n, y)\sim \mathfrak{P}}[\ell(Q_{tot}, y)]$. Then we have that the optimal \textit{expected risk minimizer} $Q_{tot}^{\star}$ as 
\begin{equation}
Q_{tot}^{\star} ={\arg \min}_{Q_{tot}\in \overline{\mathcal{Y}}} \mathcal{L}(Q_{tot}(\boldsymbol{\tau}, \boldsymbol{u})).
\label{eq5sss}
\end{equation}
Our preparation is complete, so we can now discuss the regret bound for our generalization. We state the full Theorem here.

\begin{theorem}\label{theorem}
	Let $\ell $ be 1-Lipschitz, $\delta \in (0, 1]$ and Assumption~\ref{ass1} hold with constants $\{C_1, C_2, \eta\}$. 
	Then, for $L_1$-norm models, where $\left\| \boldsymbol{a}_{ld} \right\|_1\leq B_a, 1\leq l \leq n$, and $\left\| \boldsymbol{\lambda} \right\|_1\leq B_\lambda$ where $\boldsymbol{\lambda}=\{\{\lambda_{ld}\}_{d=1}^{\rho_l}\}_{l=1}^n$, 
	there exists some absolute constants $\{C_1, C_2\}$ with probability at least $1-\delta, \delta\in(0, 1]$ that we have 
	\begin{equation}
	\mathcal{L}(\widehat{Q}_{tot}) - \mathcal{L}(Q_{tot}^{\star}) \leq  2B_\lambda \cdot \left ( \sum_{l=1}^n (B_a)^l \right )\sqrt{\frac{\log(n)}{b}}+
	\frac{C_1}{C_2}\cdot \left (\sum_{l=1}^{n} \exp (-\rho_l^\eta  )\right )+2(\sqrt{2}+1)\cdot\sqrt{\frac{\log(2/\delta )}{b}}.
	\label{eq111111}
	\end{equation}
\end{theorem}

\begin{proof}For the expected function $Q_{tot}^{\star}$, we also denote the corresponding eigenvalues as $\{\{\lambda_{ld}^{\star}\}_{d=1}^{\overline{\rho}_l}\}_{l=1}^n$ and bases as $\{\{\boldsymbol{a}^{\star}_{ld}\}_{d=1}^{\overline{\rho}_l}\}_{l=1}^n$.
	Consider the $\widetilde{Q}_{tot}\in \mathcal{Y}$ that is a ``truncated" version of the optimal  $Q_{tot}^{\star}$. Therefore, we can rewrite the regret bound as
	\begin{equation*}
	\mathcal{L}(\widehat{Q}_{tot}) - \mathcal{L}(Q_{tot}^{\star}) =
	\underbrace{\mathcal{L}(\widehat{Q}_{tot}) - \widehat{\mathcal{L}}_b(\widehat{Q}_{tot})}_{\normalsize{\textcircled{\scriptsize{1}}}} +
	\underbrace{\widehat{\mathcal{L}}_b(\widehat{Q}_{tot})- \widehat{\mathcal{L}}_b(\widetilde{Q}_{tot})}_{\leq 0 } + 
	\underbrace{\widehat{\mathcal{L}}_b(\widetilde{Q}_{tot})- \mathcal{L}(Q_{tot}^{\star})}_{\normalsize{\textcircled{\scriptsize{2}}}},
	\end{equation*}
	where the middle term $\widehat{\mathcal{L}}_b(\widehat{Q}_{tot})- \widehat{\mathcal{L}}_b(\widetilde{Q}_{tot})\leq 0$ since $\widehat{Q}_{tot}$ minimizes the empirical risk in Eq.~(\ref{eq5sss}). 
	Therefore, binding on terms $\normalsize{\textcircled{\scriptsize{1}}}$ and $\normalsize{\textcircled{\scriptsize{2}}}$ can provide us with a proof of the bound. 
	The bound for term $\normalsize{\textcircled{\scriptsize{2}}}$  is tractable, which can be proved via Lemma~\ref{lem}. Hence with probability at least $1-\delta, \delta\in(0, 1]$, we have that
	\begin{equation}
	\widehat{\mathcal{L}}_b(\widetilde{Q}_{tot})- \mathcal{L}(Q_{tot}^{\star}) \leq \sum_{l=1}^{n} \frac{C_1}{C_2}\cdot \exp(-\rho_l^\eta )+2\sqrt{\frac{\log(2/\delta )}{b}}.
	\label{lem111}
	\end{equation}
	
	Then inspired by \citet{radenovic2022neural}, we handle the term $\normalsize{\textcircled{\scriptsize{1}}}$ via bounding the Rademacher complexity~\cite{wainwright2019high}. 
	The loss function $\ell$ is Lipschitz and bounded, with probability at least $1-\delta$ for any $\delta\in(0, 1]$ over samples of length $b$. 
	These conditions allow us to apply Theorem~8 and Theorem~12 from \citet{bartlett2002rademacher}, whose proof uses McDiarmid's inequality. Thus we have that
	\begin{equation*}
	\mathcal{L}(\widehat{Q}_{tot}) - \widehat{\mathcal{L}}_b(\widehat{Q}_{tot})\leq \mathcal{R}_b(\ell\circ \mathcal{F}) + \sqrt{\frac{8\log(2/\delta )}{b}},
	\end{equation*}
	where $\mathcal{F}$ denotes the set of all joint value functions represented, i.e, $\forall Q_{tot}(\boldsymbol{\tau}, \boldsymbol{u}) \in \mathcal{F}$, and $\mathcal{R}_b$ is the empirical Rademacher complexity. According to the Theorem~12 from \citet{bartlett2002rademacher}, $\mathcal{R}_b(\ell\circ \mathcal{F})\leq 2L\cdot\mathcal{R}_b(\mathcal{F}) \leq 2L\cdot\sum_{l=1}^n \mathcal{R}_b(\mathcal{F}_l)$. Thus, we can put all the order terms together since $\ell$ is $L$-Lipschitz, and rewrite the above equation as
	\begin{equation*}
	\begin{split}
	\mathcal{L}(\widehat{Q}_{tot}) - \widehat{\mathcal{L}}_b(\widehat{Q}_{tot})
	%\leq& \mathcal{R}_b(\ell\circ \mathcal{F}) + \sqrt{\frac{8\log(4/\delta )}{b}}\\
	%\leq&  2L\cdot\mathcal{R}_b(\mathcal{F}) + \sqrt{\frac{8\log(4/\delta )}{b}}\\
	%=&  2L\cdot\mathcal{R}_b \left(\sum_{d=1}^n \mathcal{F}_b \right) + \sqrt{\frac{8\log(4/\delta )}{b}}\\
	\leq& 2L\cdot\sum_{i=1}^n \mathcal{R}_b(\mathcal{F}_i) + 2\sqrt{2}\cdot\sqrt{\frac{\log(2/\delta )}{b}},
	\end{split}
	\end{equation*}
	where $\mathcal{F}_l$ denotes the family of $F_l(\cdot)$ in the order-$l$. Therefore, since we consider the $L_1$-norm models, there exist eigenvalue $\left\| \boldsymbol{\lambda} \right\|_1\leq B_\lambda$ and base vector $\left\| \boldsymbol{a}_{ld} \right\|_1\leq B_a$, where $\forall l \in \mathcal{N}$ and $\forall d \in \{1,\cdots, \rho_l\}$. Under these constraints, the term $\normalsize{\textcircled{\scriptsize{1}}}$ can bound the empirical Rademacher complexity via Lemma~3 from \citet{dubey2022scalable} and Lemma~5.2 from \citet{massart2000some}, and we have 
	\begin{equation}
	\mathcal{L}(\widehat{Q}_{tot}) - \widehat{\mathcal{L}}_b(\widehat{Q}_{tot})\leq
	2B_\lambda \cdot \left ( \sum_{l=1}^n (B_a)^l \right )\cdot\sqrt{\frac{\log(n)}{b}} +
	2\sqrt{2}\cdot\sqrt{\frac{\log(2/\delta )}{b}}.
	\label{lem2}
	\end{equation}
	Finally, the bound for combining Eq.~(\ref{lem111})  and Eq.~(\ref{lem2}) provides us with the results of the proof.
\end{proof}

\begin{lemma}
	\label{lem}
	With probability at least $1-\delta$ for any $\delta\in(0, 1]$ and some absolute constants $\{C_1, C_2\}$, we have that
	\begin{equation*}
	\widehat{\mathcal{L}}_b(\widetilde{Q}_{tot})- \mathcal{L}(Q_{tot}^{\star}) \leq \sum_{l=1}^{n} \frac{C_1}{C_2}\cdot \exp(-\rho_l^\eta )+2\sqrt{\frac{\log(2/\delta )}{b}}.
	\end{equation*}
\end{lemma}
\begin{proof} 
	Observe,
	\begin{equation*}
	\begin{split}
	\widehat{\mathcal{L}}_b(\widetilde{Q}_{tot})- \mathcal{L}(Q_{tot}^{\star}) =
	&\widehat{\mathcal{L}}_b(\widetilde{Q}_{tot})- \mathcal{L}(\widetilde{Q}_{tot}) + 
	\mathcal{L}(\widetilde{Q}_{tot})- \mathcal{L}(Q_{tot}^{\star})\\
	\leq &\underbrace{\left|\widehat{\mathcal{L}}_b(\widetilde{Q}_{tot})- \mathcal{L}(\widetilde{Q}_{tot})\right|}_{\normalsize{\textcircled{\scriptsize{2a}}}} + 
	\underbrace{\left| \mathcal{L}(\widetilde{Q}_{tot})- \mathcal{L}(Q_{tot}^{\star})\right|}_{\normalsize{\textcircled{\scriptsize{2b}}}}
	\end{split}.
	\end{equation*}
	
	To bound $\normalsize{\textcircled{\scriptsize{2a}}}$, we have sample points $\in\mathfrak{P}$ in a batch $b$ that satisfies $\mathcal{L}(\widetilde{Q}_{tot}) = \mathbb{E}[\ell(\widetilde{Q}_{tot}, y)]$, where $0\leq\ell(\cdot, \cdot)\leq 1$. Hence we employ Azuma-Hoeffding's inequality~\cite{bercu2015concentration} and substitute the reproducing Hilbert space~(RHS) ~\cite{berlinet2011reproducing} probability with $1-\delta$, which can be rewritten in terms as 
	\begin{equation*}
	\left|\widehat{\mathcal{L}}_b(\widetilde{Q}_{tot})- \mathcal{L}(\widetilde{Q}_{tot})\right|\leq 2\sqrt{\frac{\log(2/\delta )}{b}}.
	\end{equation*}
	
	Since $\ell$ is $L$-Lipschitz, we have for some $\{\widetilde{Q}_{tot},Q^{\star}_{tot}, y\}  \in\mathcal{Y}$,
	\begin{equation*}
	\begin{split}
	\left| \ell(\widetilde{Q}_{tot}, y)- \ell(Q^{\star}_{tot} , y) \right|
	\leq &\left|L\cdot |\widetilde{Q}_{tot}- y |- L\cdot |Q^{\star}_{tot}- y| \right| \\
	= & L\cdot\left| |\widetilde{Q}_{tot}- y|- |Q^{\star}_{tot}- y  | \right| \\
	\leq & L\cdot \left|\widetilde{Q}_{tot}- Q^{\star}_{tot} \right|.
	\end{split} 
	\end{equation*}
	Thus, when $L=1$, the  bound $\normalsize{\textcircled{\scriptsize{2b}}}$ is derived as
	\begin{equation*}
	\begin{split}
	\left| \mathcal{L}(\widetilde{Q}_{tot})- \mathcal{L}(Q_{tot}^{\star})\right|
	\leq& \left|\mathbb{E}_{([Q_i]_{i=1}^n, y)\sim \mathfrak{P}}[\ell(\widetilde{Q}_{tot},y)-\ell(Q_{tot}^{\star}, y)] \right| \\
	\leq& \mathbb{E}_{([Q_i]_{i=1}^n, y)\sim \mathfrak{P}}\left [ |\ell(\widetilde{Q}_{tot},y)-\ell(Q_{tot}^{\star}, y)|\right ]
	\\
	\leq& L \cdot \mathbb{E}_{([Q_i]_{i=1}^n, y)\sim \mathfrak{P}}\left [ |\widetilde{Q}_{tot}-Q_{tot}^{\star} |\right ]
	\\
	\leq& L \cdot \sup_{Q\in \mathcal{Q}} |\widetilde{Q}_{tot}-Q_{tot}^{\star} |\\
	=& \sup_{Q\in \mathcal{Q}} |\widetilde{Q}_{tot}-Q_{tot}^{\star} |.
	\end{split} 
	\end{equation*}
	
	Observing now that $\forall Q \in \mathcal{Q}$, we have
	\begin{equation*}
	\begin{split}
	\left | \widetilde{Q}_{tot}-Q_{tot}^{\star} \right|
	= &\left|\sum_{l=1}^{n} \sum_{d=\rho_l}^{\overline{\rho}_l}  \lambda _{ld}^{\star}  \cdot  \left<\boldsymbol{a}_{ld}^{\star} , F_l({Q})\right>^l \right| \\
	\leq &\sum_{l=1}^{n} \sum_{d=\rho_l}^{\overline{\rho}_l} \left|\lambda _{ld}^{\star}  \cdot   \left<\boldsymbol{a}_{ld}^{\star} , F_l({Q})\right>^l \right| \\
	\leq & \sum_{l=1}^{n} \sum_{d=\rho_l}^{\overline{\rho}_l} \left|\lambda _{ld}^{\star}   \right|,
	\end{split} 
	\end{equation*}
	when hold on Assumption~\ref{ass1}, we have that $\lambda _{ld}= C_1 \exp(-C_2\cdot d^\eta)$ if obeys the $\eta $-exponential spectral decay. Thus, 
	\begin{equation*}
	\sum_{l=1}^{n} \sum_{d=\rho_l}^{\overline{\rho}_l} \left|\lambda _{ld}^{\star}   \right|\leq  \sum_{l=1}^n  \sum_{d=\rho_l}^{\overline{\rho}_l}   C_1 \exp(-C_2\cdot d^\eta)\leq \sum_{l=1}^{n} \int_{d=\rho_l}^{\infty}C_1 \exp(-C_2\cdot d^\eta).
	\end{equation*}
	Since $\eta\geq 1$, we can bound by the Eq.~(E.16) from~\citet{yang2020function} with the RHS as 
	\begin{equation*}
	\left | \widetilde{Q}_{tot}-Q_{tot}^{\star} \right|\leq 
	\sum_{l=1}^{n} \int_{d=\rho_l}^{\infty}C_1 \exp(-C_2\cdot d^\eta)\leq  \sum_{l=1}^{n} \frac{C_1}{C_2}\exp(-\rho_l^\eta ).
	\end{equation*}
	Therefore, we finish the proof of Lemma~\ref{lem}.
	
\end{proof}

%
%\begin{lemma}
%	\label{lem2}
%	xxxx
%	\begin{equation}
%	xxxx
%	\end{equation}
%\end{lemma}
%\begin{proof} 
%\end{proof}

\section{Variational Auto-Encoder Background}\label{app3}
A variational auto-encoder (VAE)~\cite{sohn2015learning} is a popular generative model to learn an attention mask, e.g., U-Net~\cite{ronneberger2015u} for semantic segmentation.
VAE aims to maximize the marginal log-likelihood $\log p({T})=\sum_{j=1}^{b}\log p(\tau^j)$, where ${T}=[\tau^j]_{j=1}^b\in \mathcal{T}$ denotes the set of local action-observation histories from $\mathcal{B}$, and it is common to replace the optimized variational lower-bound as
\begin{equation*}
\log p({T})\geq \mathbb{E}_{q({T}|z)}\left [ \log p({T}|z) \right ]+D_{\text{KL}}(q(z|{T})||p(z)),
\end{equation*}
where $p(z)$ generally is a multivariate normal distribution $\mathcal{N}(0, I)$ to represent the prior. We define the posterior $q(z|{T}) = \mathcal{N}(z|\mu, \sigma^2({T})I)$ as the encoder $E_{\omega_1}$ and $p({T}|z)$ as the decoder $D_{\omega_2}$. 
It is understood that given a sample $\tau$ is fed into the VAE to produce a latent semantic vector $z$, and then this vector is reconstructed into the desired sample by training. 
To apply gradient descent on the variational lower-bound, we allow the re-parametrization trick~\cite{rezende2014stochastic} to train on a reconstruction loss with a KL-divergence as
\begin{equation*}
\mathbb{E}_{z\sim \mathcal{N}(\mu, \sigma )}\left [ f(z) \right ] = \mathbb{E}_{\nu  \sim\mathcal{N}(0, I)} \left [ f(\mu+\sigma\nu  ) \right ].
\end{equation*}
Thus $\mu$ and $\sigma$ can be represented by deterministic functions, allowing for back-propagation.

\section{Pseudo Code}\label{pseudo}

\begin{algorithm}[H]
	%\textsl{}\setstretch{1.8}
	\renewcommand{\algorithmicrequire}{\textbf{Input:}}
	\renewcommand{\algorithmicensure}{\textbf{Output:}}
	\caption{Neural Attention Additive Q-learning}
	\label{alg1}
	\begin{algorithmic}
		\STATE Initialize a set of agents $\mathcal{N}=\{1,2,\cdots, n\}$
		\STATE Initialize networks of local agents $Q_i(\tau_i, u_i; \theta)$ and target networks
		$Q_i(\tau_i', u_i';\hat{\theta})$,  $G_{\hat{\omega}}$ with $\hat{\theta} \leftarrow \theta$
		\STATE	Initialize a VAE $G_{\omega} = \left\{E_{\omega_1},  D_{\omega_2} \right\} $ with parameters  $\omega$
		\STATE Initialize a replay buffer  $\mathcal{B}$ for storing episodes
		\REPEAT
		\STATE Initialize a history embedding $h_i^0$ and an action vector $u^0_i$ for each agent 
		\STATE Observe each agent's partial observation $\left[o^1_{i}\right]_{i=1}^{n}$
		
		\FOR{$t=1:T$}
		\STATE Get $\tau_i^t=\left\{o^t_{i}, h^{t-1}_{i}\right\}$ for each agent and calculate the individual value function $Q_i(\tau^t_i, u^{t-1}_i)$
		\STATE  Get the hidden state $h^t_i$ and select action $u^t_i$ via value function with probability $\varepsilon$ exploration
		\STATE Unsampled $n$ identity semantic masks $\left[\mathcal{M}_i \sim G_\omega(h_i^t)\right]_{i=1}^n$ as an interpretation
		\STATE Execute $u^t_i$ to receive the reward $r^t$, next state $\boldsymbol{s}^{t+1}$
		\ENDFOR
		
		\STATE Store the episode trajectory to $\mathcal{B}$
		\STATE Sample a batch of episodes trajectories with batch size $b$ from $\mathcal{B}$
		
		\FOR{$t=1:T$}
		\STATE Calculate $\mu, \sigma = E_{\omega_1}(\tau_i^t)$ and identity semantics $\boldsymbol{z}=[z_i \sim  \mathcal{N}(\mu, \sigma)]^n_{i=1}$
		\STATE Get $\widetilde{o_i} = \mathcal{M}_i\odot o_i$ and calculate $\mathcal{L}_{G_{\omega}}$ via Eq.~(\ref{eq3})
		\STATE Get the attention weight $\alpha_k(\boldsymbol{z}, \boldsymbol{s} )$  by the intervention function in Eq.~(\ref{eq4})
		\STATE Calculate the joint value function within order-$2$ interactions via Eq.~(\ref{eq5}) 
		%		\STATE Calculate $y$ using target networks via  Eq.~\ref{eq6}
		\ENDFOR
		
		\STATE Construct the loss function defined in Eq.~(\ref{eq8})
		\STATE Update $\omega$ and $\theta$ by minimizing the above loss
		\STATE Periodically update  $\hat{\theta} \leftarrow \theta$ 
		\UNTIL $Q_i(\tau_i, u_i; \theta)$ converges
	\end{algorithmic}  
\end{algorithm}

\section{Related work}
\textbf{Value Decomposition in MARL.} 
Since the joint action space grows exponentially in proportion to the number of participating agents~\cite{yang2018mean}, the centralized training and decentralized execution (CTDE)~\cite{oliehoek2008optimal} paradigm is proposed to relieve this issue and become a mainstream framework in MARL. 
One of the crucial challenges in CTDE is credit assignment, which aims to infer how much each agent contributes to the overall success. 
Under the CTDE framework, VDN~\cite{sunehag2017value} assumes that any joint action-value function can be decomposed into a linear summation of individual value functions.
Nevertheless, this equivalent factorization limits the credit assignment of the global Q-value. 
To mitigate this issue, some implicit credit assignment methods, e.g., QMIX~\cite{son2019qtran} and QTRAN~\cite{wang2020qplex}, represent the joint value function into a richer family for value decomposition with complex nonlinear transformation function. 
Further, Weighted QMIX~\cite{rashid2020weighted} proposes a weighted projection to decompose the joint action-value function, and PMIC~\cite{PMIC} utilizes more effective mutual information to collaborate better.
However, these methods neglect causal explanations in credit assignment, which may be unreasonable since suboptimal actions lack an explicit reasoning mechanism.
They entangle the interactions at temporal hidden layers for credit assignment.
Thus, recent works~\cite{wang2021shaq, li2021shapley} apply the Shapley theory to trustworthiness for inferring the credits, where fairness is achieved by considering the incremental marginal contribution of one of the agents. 
These methods fail to interpret the impact of agent observation on decision-making or explicitly present how they cooperate with each other.
Whereas glass-box models in MARL, e.g., mixture soft decision trees~\cite{liu2022mixrts} and visual perception~\cite{blumenkamp2021emergence}, do not achieve exciting performance. 
To resolve these problems, we propose a novel interpretable value decomposition method in this paper.

\textbf{Generalized Additive Models.} GAMs are generally regarded as powerful inherently-interpretable models in the machine learning community~\cite{hastie1986generalized}. 
It independently learns a shape function for each feature and sums the outputs of these functions to obtain the final model prediction.
Previous work~\cite{lou2013accurate} has found that standard forms of GAMs are limited in their representational power due to the absence of learning interactions between inherent features.
As an improvement, \citet{lou2013accurate} proposed
G$\text{A}^\text{2}$M that incorporates the complexity of pairwise interactions into GAMs.
To improve stability and performance, different variants of shape functions in GAMs have been investigated, including deep neural networks~\cite{agarwal2021neural}, polynomial kernel models~\cite{dubey2022scalable}, and oblivious decision trees~\cite{chang2021node}.
Further, NIT~\cite{tsang2018neural} and pureGAM~\cite{sun2022puregam} reduce complexity by adding constraint terms, achieving increased interpretability.
Our work falls under the umbrella of the GAM family.
We are the first to develop GAMs in value-based MARL by utilizing them to disentangle the joint action-value function across different interactions, thereby obtaining intrinsic and interpretable higher-order shape functions of the agents.

%rezende2014stochastic\mathcal{N}(, \sigma)
\section{Experimental Details}\label{app4444}

\subsection{Benchmarks and Settings}
In our paper, we introduce two types of testing benchmarks as shown in \textit{Figure~\ref{stab}}, including Level Based Foraging~(LBF) and StarCraft Multi-Agent Challenge~(SMAC). In this section, we will describe the details and settings of these benchmarks.

\begin{figure*}[ht]
	\centering %图片全局居中
	%	\subfigbottomskip=0pt %两行子图之间的行间距
	%	\subfigcapskip=-5pt %设置子图与子标题之间的距离
	\subfigure[Level Based Foraging]{
		\includegraphics[width=.3\linewidth]{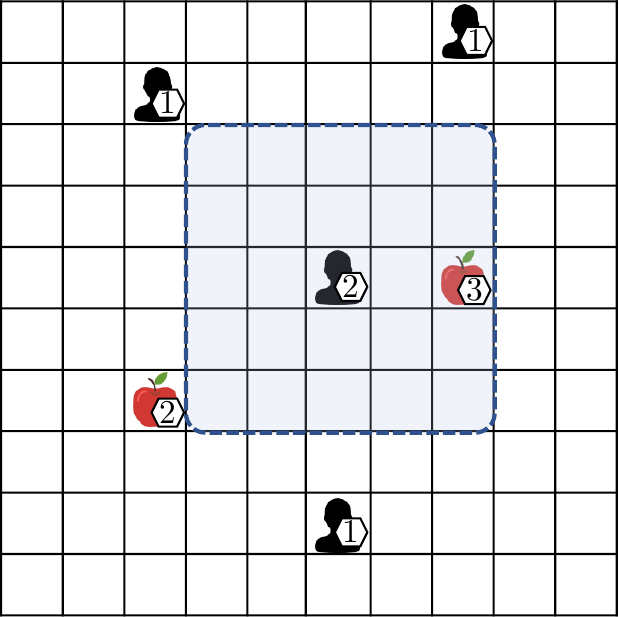}}
	\hspace{.5in}
	\subfigure[StarCraft Multi-Agent Challenge]{
		\includegraphics[width=.43\linewidth]{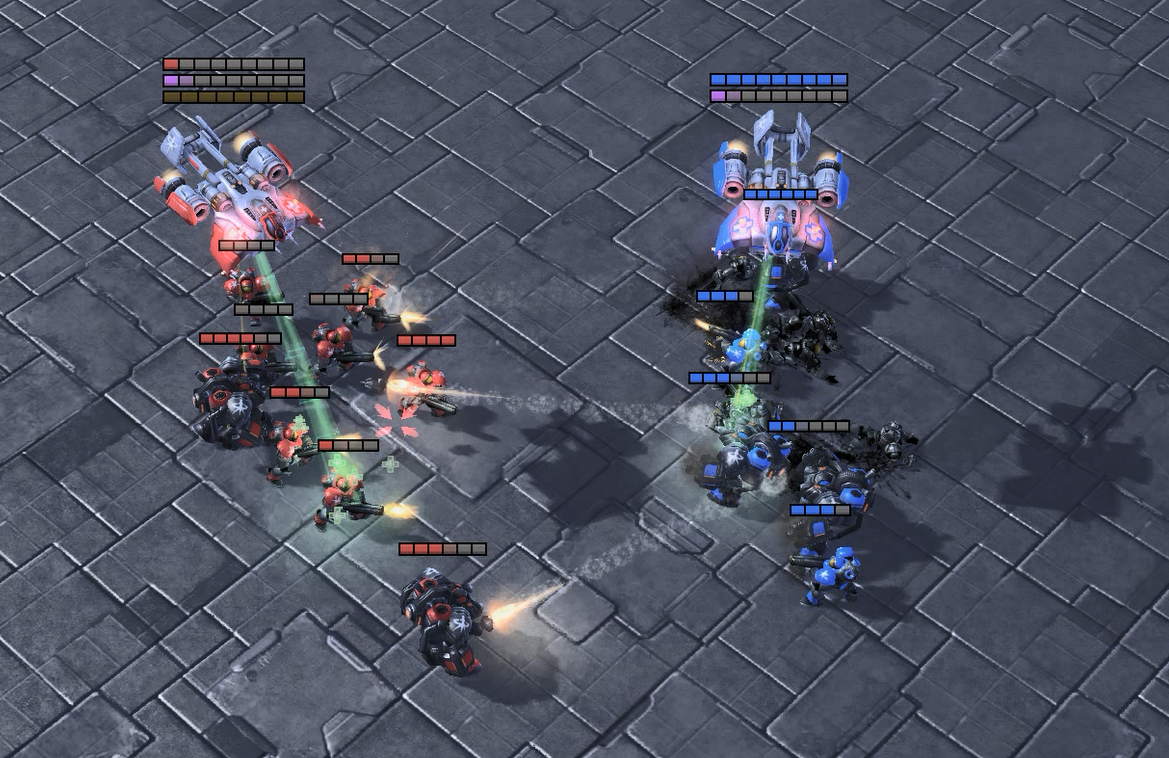}}
	\caption{Two benchmarks used in our experiments.}
	\label{stab}
\end{figure*}

\textbf{Level Based Foraging.} \citet{christianos2020shared} first uses this implementation of
LBF. This environment is a mixed game of cooperation and competition. Specifically, agents are placed in a $10\times 10$ grid world and each one is assigned a level.
The observation of an agent has a $5\times 5$ field of view around it.
Their goal is to eat food that is also randomly scattered. 
Only if the total level of the agents involved in eating is equal to or higher than the food level does the agents receive a positive reward, which is the normalized food level.
Furthermore, we set the penalty reward for movement to $-0.002$.
On this basis, we use two task instances with different configurations, of which one is $2$ food with $4$ agents, and $3$ food with $3$ agents. 
We give other experimental settings in \textit{Table~\ref{table4}}.

\begin{table}[t]
	\caption{Experimental settings of Level Based Foraging.}\label{table4}
	\vskip 0.15in
	%\fontsize{9}{10}\selectfont 
	\begin{center}
		\begin{small}
			\begin{sc}
				\begin{tabular}{lcl}
					\toprule
					Hyperparameter & Value & Description \\
					\midrule
					Max player level & 3 & Maximum agent level attribute\\
					max episode length & 50 & Maximum Timesteps per episode\\
					batch size & 32 &  Number of episodes per update\\
					test interval & 10,000 &  Frequency of evaluating performance\\
					test episodes & 32 & number of episodes to test\\
					Replay batch size    & 5000 & Maximum number of episodes stored in memory\\
					Discount factor $\gamma $  &0.99 & Degree of impact of future rewards \\
					Total  timesteps  & 1,050,000 &  Number of training steps  \\
					start $\varepsilon$   & 1.0   & the start $\varepsilon$ value to explore\\
					finish $\varepsilon$   & 0.05   & the finish $\varepsilon$ value to explore\\
					Anneal Steps for $\varepsilon$  &50, 000 & number of steps of linear annealing\\
					Target update interval  &  200  &  the target network update cycle\\
					\bottomrule
				\end{tabular}
			\end{sc}
		\end{small}
	\end{center}
	\vskip -0.1in
\end{table}

\label{4fafas}
\textbf{StarCraft Multi-Agent Challenge.} The SMAC~\cite{samvelyan2019starcraft} is one of the most popular multi-agent environments to test the performance of MARL algorithms. All algorithm implementations are based on StarCraft II (SC2.4.10 version) unit micromanagement tasks, and note that results from different versions are not comparable. 
We set the built-in AI difficulty of all enemy units by configuring $\text{difficulty=7}$, and all allied units are controlled by the corresponding RL algorithm. 
The allies need to learn a series of strategies to defeat all the enemies and win within the specified exploration length.
In this paper, we evaluate all algorithms on 12 challenging combat scenarios in SMAC, and \textit{Table~\ref{bibisadadd}} presents a brief introduction of these scenarios and the maximum training step.
% Besides, the specific environmental settings follow the original setups, which are described in \textit{Table~\ref{bibdasdaid}}.
Furthermore, the specific environmental settings adhere to the original setups, as described in \textit{Table~\ref{bibdasdaid}}.

\begin{table}[t]
	\caption{Introduction of scenarios in SMAC benchmark.}\label{bibisadadd}
	\vskip 0.15in
	\begin{center}
		\begin{small}
			\begin{sc}
				\begin{tabular}{cllcc}
					\toprule
					Map name & Ally Units & Enemy Units & Total  timesteps & Scenario Type\\
					\midrule
					$8m$ &8 Marines & 8 Marines& $2M$ &  Easy\\
					$2s3z$ &2 Stalkers, 3 Zealots& 2 Stalkers, 3 Zealots&$2M$ &Easy\\
					$2s\_vs\_1sc$ &2 Stalkers&1 Spine Crawler &$2M$ &Easy\\
					\hline
					$3s5z$ &3 Stalkers, 5 Zealots&3 Stalkers, 5 Zealots&$2M$ &Hard\\
					$3s\_vs\_5z$ &3 Stalkers&5 Zealots&$2M$ &Hard\\
					$2c\_vs\_64zg$ &2 Colossi&64 Zerglings&$2M$ &Hard\\
					$5m\_vs\_6m$ &5 Marines & 6 Marines&$2M$ &Hard\\
					$8m\_vs\_9m$ &8 Marines & 9 Marines&$2M$ &Hard\\
					\hline
					\multirow{2}{*}{$\textit{MMM}2$}& 1 Medivac, 2 Marauders,&1 Medivac, 	3 Marauders,&\multirow{2}{*}{ $2M$ }&\multirow{2}{*}{Super hard}\\
					&and  7 Marines&and 8 Marines& & \\
					$3s5z\_vs\_3s6z$ &3 Stalkers, 5 Zealots&3 Stalkers, 6 Zealots&$5M$ &Super hard\\
					$corridor$ &6 Zealots&24 Zerglings&$5M$ &Super hard\\
					$6h\_vs\_8z$ &6 Hydralisks& 8 Zealots&$5M$ &Super hard\\
					\bottomrule
				\end{tabular}
			\end{sc}
		\end{small}
	\end{center}
	\vskip -0.1in
\end{table}

\begin{table}[ht]
	\caption{Experimental settings of StarCraft Multi-Agent Challenge.}\label{bibdasdaid}
	\vskip 0.15in
	%\fontsize{9}{10}\selectfont 
	\begin{center}
		\begin{small}
			\begin{sc}
				\begin{tabular}{lcl}
					\toprule
					Hyperparameter & Value & Description \\
					\midrule
					difficulty & 7 & Enemy units with built-in AI difficulty\\
					batch size & 32 &  Number of episodes per update\\
					test interval & 10,000 &  Frequency of evaluating performance\\
					test episodes & 32 & number of episodes to test\\
					Replay batch size    & 5000 & Maximum number of episodes stored in memory\\
					Discount factor $\gamma $  &0.99 & Degree of impact of future rewards \\
					start $\varepsilon$   & 1.0   & the start $\varepsilon$ value to explore\\
					finish $\varepsilon$   & 0.05   & the finish $\varepsilon$ value to explore\\
					Anneal Steps for easy \& hard  &50,000& number of steps of linear annealing $\varepsilon$\\
					Anneal Steps for super hard  &100,000& number of steps of linear annealing $\varepsilon$\\
					Target update interval  &  200  &  the target network update cycle\\
					\bottomrule
				\end{tabular}
			\end{sc}
		\end{small}
	\end{center}
	\vskip -0.1in
\end{table}

\begin{table}[ht]
	\caption{The specific structure of the shape function.}\label{bibdasdadfafaid}
	\vskip 0.15in
	%\fontsize{9}{10}\selectfont 
	\begin{center}
		\begin{small}
			\begin{sc}
				\begin{tabular}{ll}
					\toprule
					No. &Structure \\
					\midrule
					1st layer& [abs(linear.weight), Linear(order number, 8), elu] \\
					2nd layer& [abs(linear.weight), Linear(8, 4), elu] \\
					3rd layer& [abs(linear.weight), Linear(4, 1)] \\
					\bottomrule
				\end{tabular}
			\end{sc}
		\end{small}
	\end{center}
	\vskip -0.1in
\end{table}

\subsection{Hyperparameters of Baselines} 
We compare our method against nine popular value-based baselines, including VDN~\cite{sunehag2017value}, QMIX~\cite{rashid2018qmix}, QTRAN~\cite{son2019qtran}, Qatten~\cite{yang2020qatten}, QPLEX~\cite{wang2020qplex}, Weighted QMIX (mainly OW-QMIX, and we rename it WQMIX in our experiments)~\cite{rashid2020weighted}, CDS\footnote{The code of CDS is from \url{https://github.com/lich14/CDS}.}~\cite{li2021celebrating}, DVD~\cite{li2022deconfounded}, and SHAQ\footnote{The code of SHAQ is from \url{https://github.com/hsvgbkhgbv/shapley-q-learning}.}~\cite{wang2021shaq}, whereas the implementation of baselines is based on PyMARL\footnote{The source code of implementations is from \url{https://github.com/oxwhirl/wqmix}.}.  
All hyperparameters follow the code provided by the authors, and are maintained at a learning rate of 0.0005 by the RMSprop optimizer.
Note that the learning rate of SHAQ is fine-tuned to each different scenario, which is unfair to the other baselines, hence the hyperparameters are set identically to others.

\subsection{Hyperparameters of N$\text{A}^\text{2}$Q}
In this paper, we utilize a recurrent style local Q-network with its default hyperparameters, specifically, the individual Q-function $Q_i(\tau_i, u_i)$ contains a GRU layer with a 64-dimensional hidden state and a ReLU activation layer.
The optimization for individual Q-functions is conducted using RMSprop with weight decay and a learning rate of 0.0005.
Regarding the generative model $G_{\omega}$, both encoder and decoder are comprised of two fully connected layers with a 32-dimensional hidden state,  optimizing the learnable parameters by Adam with a learning rate of 0.0005.
Additionally, we set the weight $\beta$ of the loss to 0.1.
In the mixing network, we employ a small dimensional MLP for each shape function $f_k$ in order-1 and order-2, whose details are shown in \textit{Table~\ref{bibdasdadfafaid}}. 
Finally, for the attention mechanism, we set the hidden layer size to 64 for $\boldsymbol{w}_s$ and $\boldsymbol{w}_z$.

\subsection{Infrastructure}
Experiments are performed on an NVIDIA RTX 3080Ti GPU and an Intel I9-12900k CPU. We train our approach to run from 1 to 20 hours per scenario, depending on the complexity and length of the episode for each scenario.

\section{Interpretability on LBF}\label{LBF}
\textit{Figure~\ref{overview_lbf}} demonstrates the contribution of agents and sub-teams on an episode in the LBF task, as well as showing the agent's corresponding mask.
It is clear that N$\text{A}^\text{2}$Q accurately models the contribution of any agent or coalition of agents to the overall success.
Furthermore, unsampled individual semantics can help us diagnose in a more interpretable way the relative importance of individual agent masks to relevant observations in the decision-making process.
\begin{figure}[H]
	\centering %图片全局居中
	%	\subfigbottomskip=0pt %两行子图之间的行间距
	%	\subfigcapskip=-5pt %设置子图与子标题之间的距离
	\subfigure[step = 1]{
		\includegraphics[width=.24\linewidth]{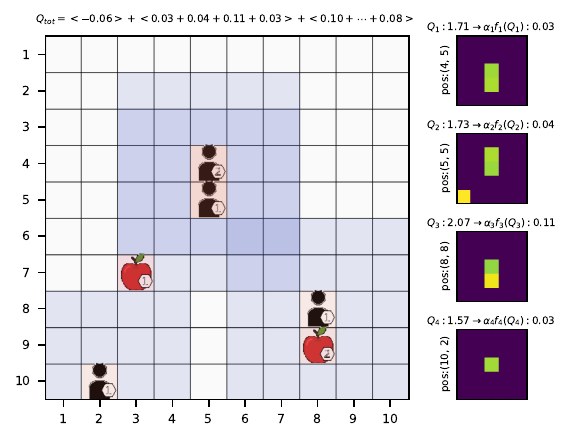}}
	\subfigure[step = 2]{
		\includegraphics[width=.24\linewidth]{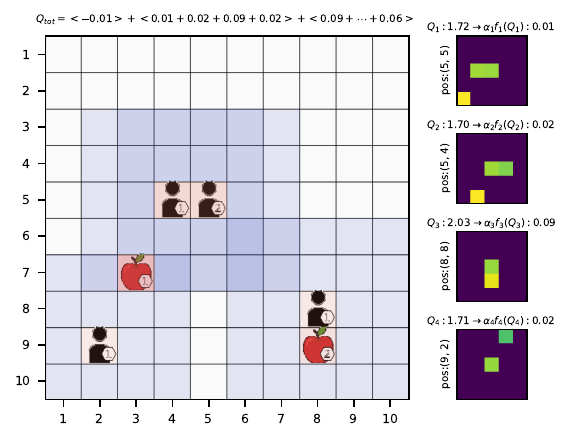}}
	\subfigure[step = 3]{
		\includegraphics[width=.24\linewidth]{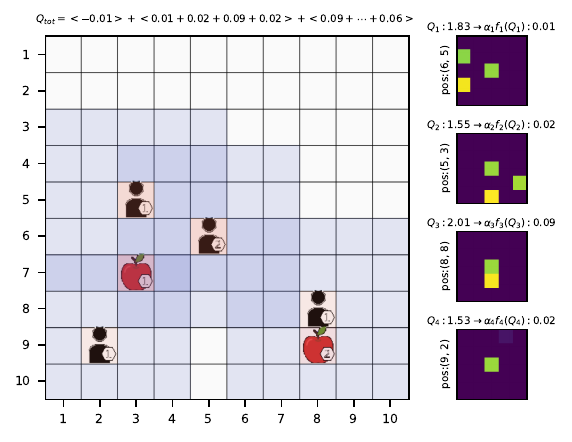}}	\subfigure[step = 4 (eating)]{
		\includegraphics[width=.24\linewidth]{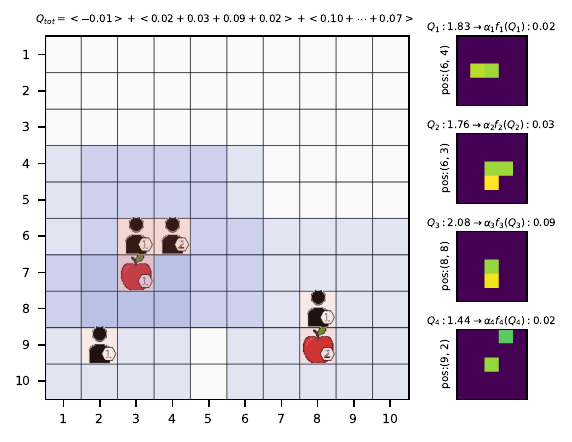}}\\
	\vskip -0.1in
	\subfigure[step = 5]{
		\includegraphics[width=.24\linewidth]{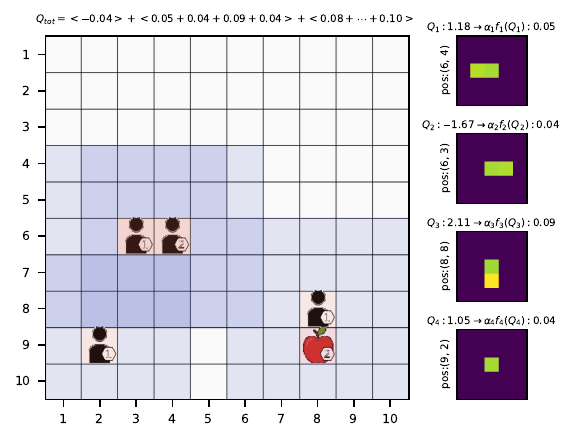}}
	\subfigure[step = 6]{
		\includegraphics[width=.24\linewidth]{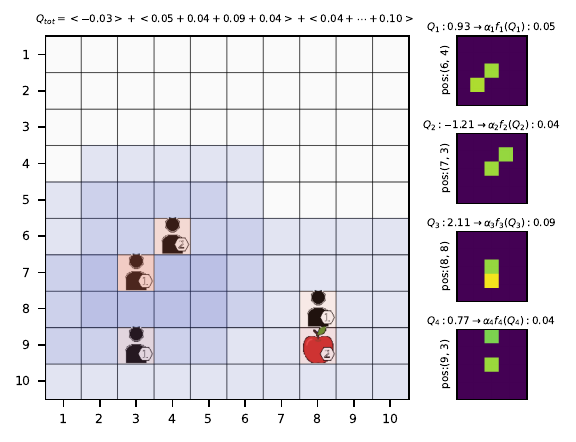}}
	\subfigure[step = 7]{
		\includegraphics[width=.24\linewidth]{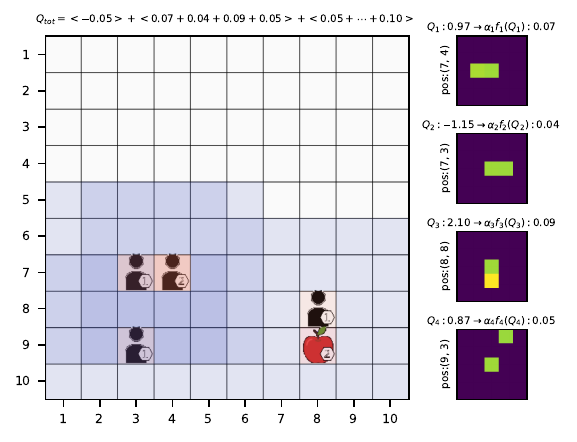}}
	\subfigure[step = 8]{
		\includegraphics[width=.24\linewidth]{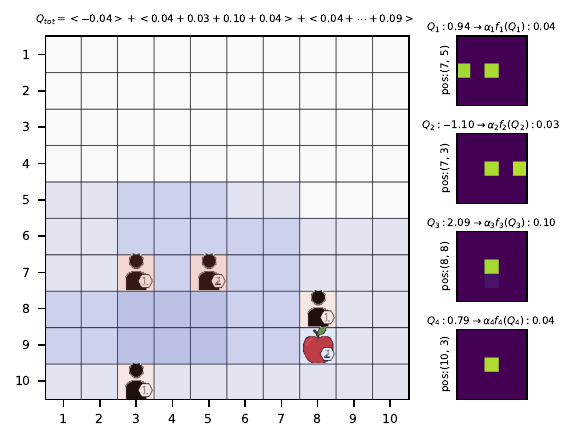}}\\
	\vskip -0.1in
	\subfigure[step = 9]{
		\includegraphics[width=.24\linewidth]{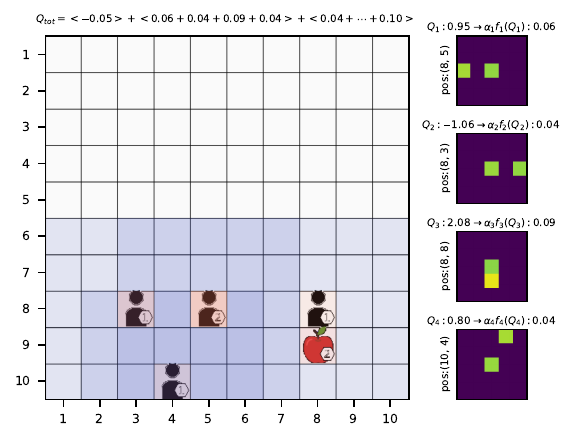}}
	\subfigure[step = 10]{
		\includegraphics[width=.24\linewidth]{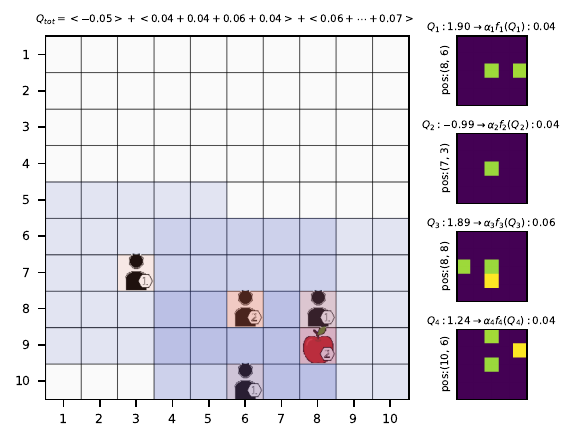}}
	\subfigure[step = 11]{
		\includegraphics[width=.24\linewidth]{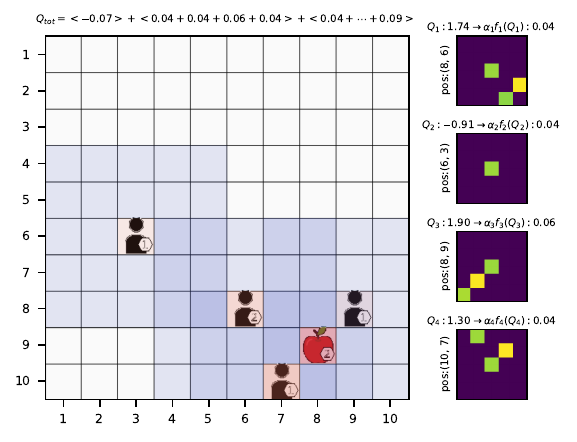}}
	\subfigure[step = 12 (eating)]{
		\includegraphics[width=.24\linewidth]{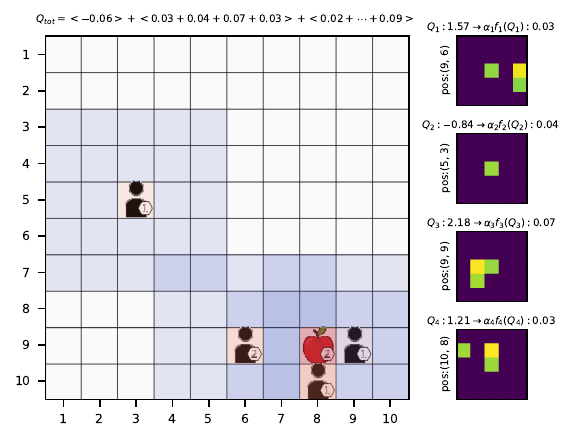}}
	\vskip -0.1in
	\caption{Visualization of the agent's mask on an episode,  and the title indicates the contribution of each individual and agent alliance. The highlighted areas are the important regions for making decisions. 
		As expected, when the environment changes, the attention and contribution of the agents also change accordingly.
	}\label{overview_lbf}
\end{figure}

\begin{figure}[ht]
	%	\vskip 0.2in
	\begin{center}
		\subfigure[Agent properties and corresponding mask values]{		
			\includegraphics[width=.81\linewidth]{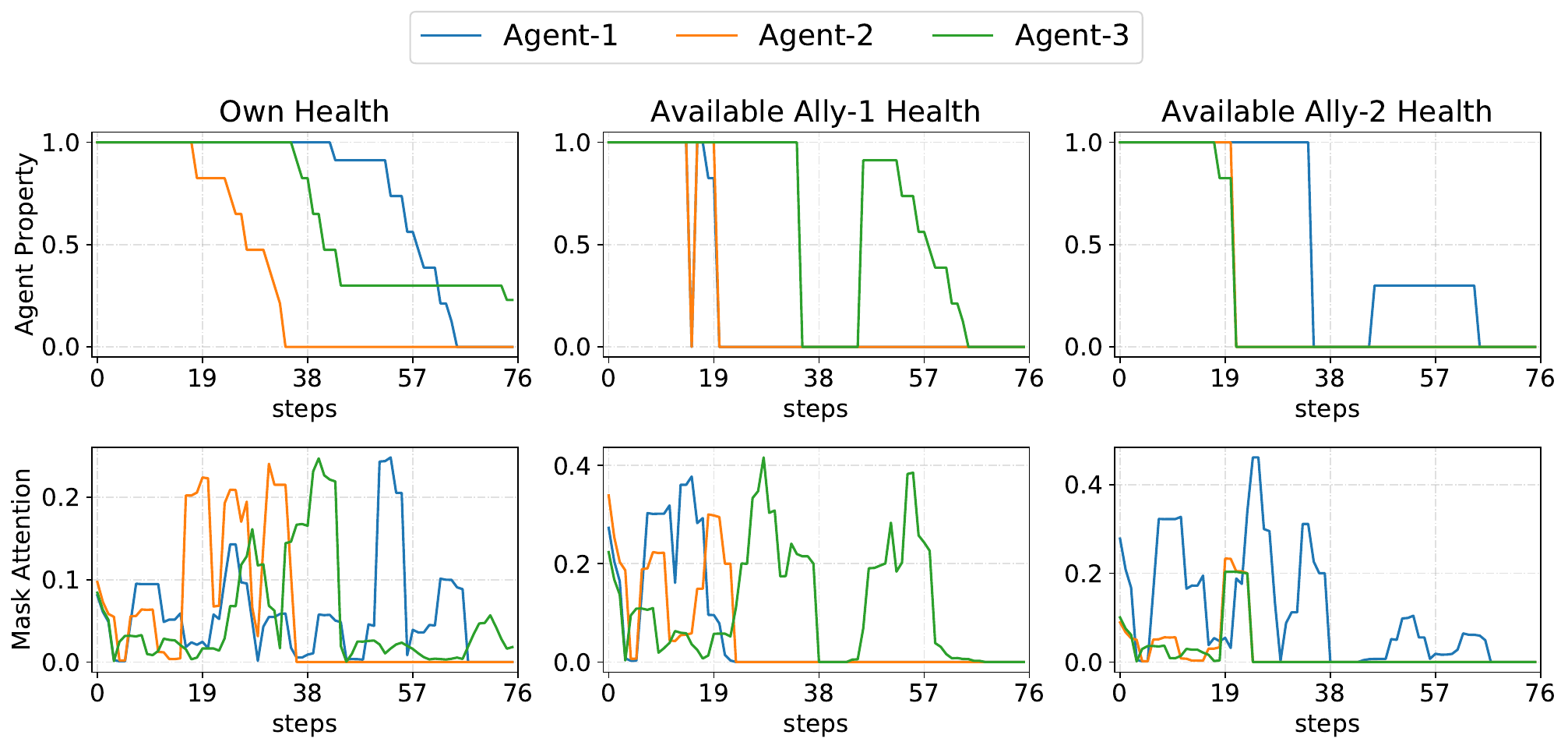}}
		\subfigure[Agent contributions]{
			\includegraphics[width=.17\linewidth]{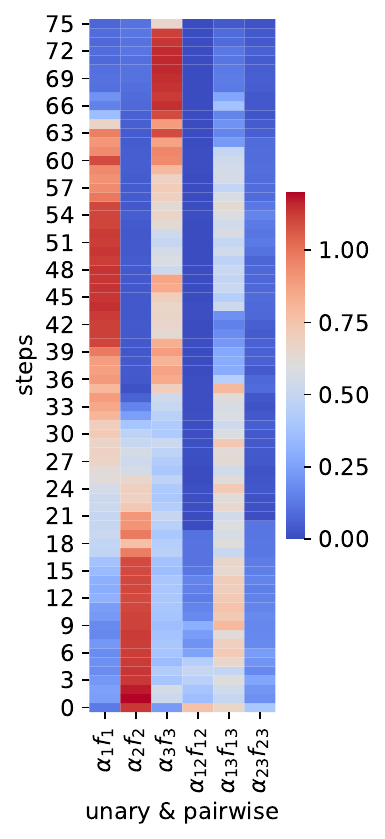}}
		\caption{Visualization of property semantics and agent contributions on  the 3s\_vs\_5z scenario.}\label{smacccccc1}
	\end{center}
\end{figure}

\begin{figure}[ht]
	\begin{center}
		\centerline{\includegraphics[width=\linewidth]{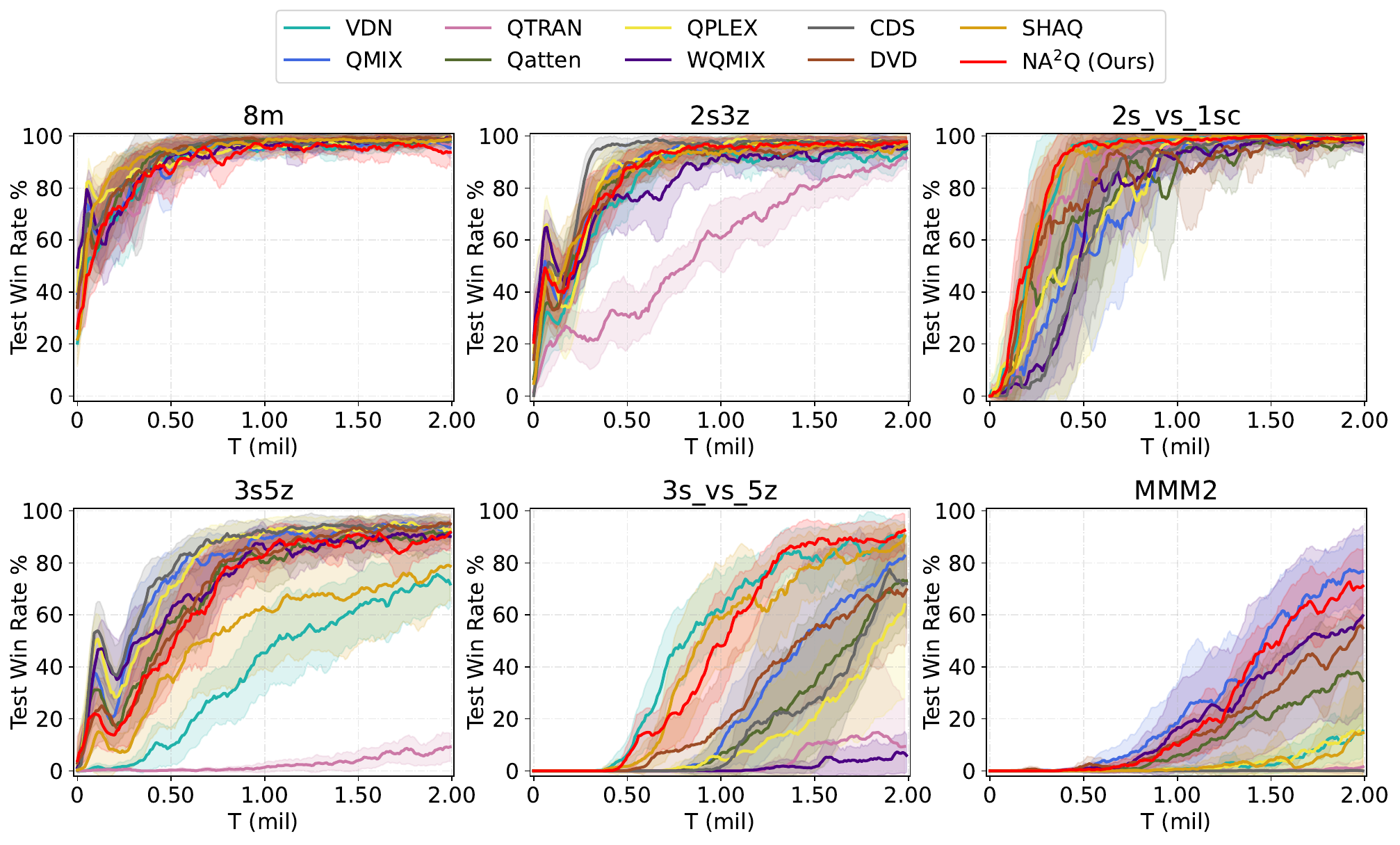}}
		\caption{Test win rate \% for six extra scenarios of SMAC benchmark.}
		\label{overview_results_app}
	\end{center}
\end{figure}

\section{Performance Results on Extra SMAC Maps}\label{Extra}

To thoroughly compare the performance of our method against the baselines, we experiment with six extra scenarios in \textit{Figure~\ref{overview_results_app}} on SMAC, including 8m, 2s3z, 2s\_vs\_1sc, 3s5z, 3s\_vs\_5z, and MMM2. 
The parameter settings are in accordance with the previous experiments.
%From the \textit{Figure~\ref{overview_results_app}}, we can observe that N$\text{A}^\text{2}$Q still achieves impressive results across these six scenarios.
It is obvious that N$\text{A}^\text{2}$Q still achieves impressive results on these six scenarios.

\section{Additional Interpretability on SMAC}\label{maskSMAC}

To further clarify the interpretability of N$\text{A}^\text{2}$Q, we select three properties related to the health of the agents to represent identity semantics, including own health, available Ally-1 health, and available Ally-2 health, and display the contribution of the corresponding agent on an episode. 
As shown in \textit{Figure~\ref{smacccccc1}(a)}, the horizontal coordinate represents the number of steps on the episode, and the two vertical coordinates represent corresponding properties and semantic mask values, respectively. 
We find that the importance of the mask increases when the observed agent is harmed.
Specifically, the teams are attacked with the sequence of Agent-2, Agent-3, and Agent-1, and the importance of their features peaked, respectively.
Also, the corresponding mask is elevated when the visible ally receives damage.
At the same time, we visualize the agent contributions to the unary and pairwise shape functions as shown in \textit{Figure~\ref{smacccccc1}(b)}, where the steps increase from bottom to top and the horizontal ordination indicates the contribution id.
The results show that the agents have different sensitivities at different stages of the battle.
For example, Agent-2 performs a kiting operation, causing it to have a high contribution at the beginning stage. 
Meanwhile, Agent-1 and Agent-3 engage in cooperative attacks, resulting in higher contributions from sub-teams than from individual agents.
In the later stages, agents are attacked separately, leading to higher contributions from individuals. 
Notably, the earlier death of Agent-2 leads to the pairwise shape functions associated with it remaining at depressed values.
In summary, the N$\text{A}^\text{2}$Q can understand complex observations by diagnosing identity semantics and better explain the sub-spaces within order-$2$ interactions for the decomposition of the joint action-value function.

\end{document}